\documentclass{article}

\usepackage[sectionbib]{natbib}
\usepackage{chapterbib}

\usepackage{lmodern}
\usepackage{tgpagella}
\usepackage{helvet}
\usepackage{setspace}
\usepackage[indent]{parskip}

\usepackage{colortbl}    
\usepackage[utf8]{inputenc}
\usepackage[T1]{fontenc}
\usepackage[english,french]{babel}
\usepackage{longtable}   
\usepackage{hyperref}
\usepackage{amssymb}
\usepackage{pifont}
\usepackage[babel=true]{microtype}
\usepackage{geometry}
\usepackage{amsmath,amsfonts,amssymb,amsthm}
\usepackage{mathtools}
\usepackage{mathrsfs}
\usepackage{nccmath}
\usepackage{dsfont}
\usepackage{bbm}
\usepackage{cancel}

\renewcommand{\epsilon}{\varepsilon}
\renewcommand{\leq}{\leqslant}
\renewcommand{\geq}{\geqslant}

\newtheorem{proposition}{Proposition}[section]
\newtheorem{definition}{Definition}[section]

\usepackage{etoolbox}
\usepackage{hyperref}

\makeatletter
\newcommand{\assumptionname}[1]{%
  \gdef\@currentlabelname{#1}%
}
\makeatother

\renewcommand{\eqref}[1]{eq.\,\ref{#1}}

\makeatletter
\newcounter{assump}
\renewcommand{\theassump}{\arabic{assump}}

\newcommand{\assumplabel}[2]{%
  \phantomsection%
  \def\@currentlabelname{#2}%
  \label{#1}%
}

\newcommand{\assumplabelstep}[2]{%
  \refstepcounter{assump}%
  \phantomsection%
  \def\@currentlabelname{#2}%
  \label{#1}%
  \textbf{(\theassump)}\ %
}

\usepackage{graphicx}

\usepackage{wrapfig}
\usepackage{pdflscape}
\usepackage{booktabs}
\usepackage{multirow}
\usepackage{array}
\usepackage{tabularx}
\usepackage{longtable}
\usepackage{makecell}
\usepackage[export]{adjustbox}
\usepackage{pgfplots}
\pgfplotsset{compat=1.18}

\newcolumntype{P}[1]{>{\RaggedRight\arraybackslash}p{#1}}
\newcolumntype{X}[1]{>{\RaggedRight\arraybackslash}p{#1}}
\newcolumntype{M}[1]{>{\centering\arraybackslash}m{#1}}

\usepackage[
  font={small,sf},
  labelfont=bf,
  format=plain,
  margin=0pt
]{caption}

\usepackage[list=true]{subcaption}
\usepackage{xcolor}

\definecolor{hardgreen}{RGB}{102,194,165}
\definecolor{softgreen}{RGB}{166,216,84}
\definecolor{myyellow}{RGB}{255,217,47}
\definecolor{mypink}{RGB}{231,138,195}
\definecolor{myblue}{RGB}{141,160,203}
\definecolor{mygray}{RGB}{179,179,179}
\definecolor{gold}{RGB}{229,196,148}
\definecolor{UMColor}{RGB}{255,86,96}

\definecolor{grey_cbf1}{HTML}{999999}
\definecolor{orange_cbf1}{HTML}{E69F00}
\definecolor{blue_cbf1}{HTML}{56B4E9}
\definecolor{green_cbf1}{HTML}{009E73}
\definecolor{yellow_cbf1}{HTML}{F0E442}
\definecolor{navy_cbf1}{HTML}{0072B2}
\definecolor{red_cbf1}{HTML}{D55E00}
\definecolor{purple_cbf1}{HTML}{CC79A7}

\definecolor{black_cbf2}{HTML}{000000}
\definecolor{orange_cbf2}{HTML}{E69F00}
\definecolor{blue_cbf2}{HTML}{56B4E9}
\definecolor{green_cbf2}{HTML}{009E73}
\definecolor{yellow_cbf2}{HTML}{F0E442}
\definecolor{navy_cbf2}{HTML}{0072B2}
\definecolor{red_cbf2}{HTML}{D55E00}
\definecolor{purple_cbf2}{HTML}{CC79A7}
\usepackage{hyperref}
\hypersetup{
  colorlinks   = true,
  linkcolor    = black,
  citecolor    = brown,
  urlcolor     = brown,
  breaklinks   = true
}

\usepackage{cleveref}

\usepackage[noend]{algpseudocode}
\usepackage[titlenumbered,ruled,noend,vlined,norelsize]{algorithm2e}
\newcounter{algorithmct}

\usepackage{listings}
\usepackage[inline]{enumitem}
\usepackage{siunitx}
\usepackage{numprint}

\usepackage{tikz}
\usetikzlibrary{
  arrows,
  arrows.meta,
  shapes,
  shapes.geometric,
  positioning,
  shadows,
  trees,
  mindmap,
  backgrounds,
  tikzmark,
  calc
}
\usepackage[edges]{forest}
\usepackage{xkcdcolors}

\tikzset{
  > = stealth,
  every node/.append style = {
    draw = none,
    text = black
  },
  every path/.append style = {
    arrows = ->
  },
  hidden/.style = {
    draw = black,
    shape = circle,
    inner sep = 2pt
  }
}

\tikzstyle{node.box}  = [rectangle, rounded corners, minimum width=3cm, minimum height=1cm, align=center, draw=black]
\tikzstyle{node.text} = [minimum width=3cm, minimum height=1cm, align=center]
\tikzstyle{arrow}     = [thick,->,>=stealth]

\colorlet{linecol}{black!75}
\usepackage[most]{tcolorbox}

\newtcbox{\gfmath}{on line,
  boxrule=0.8pt,
  arc=2.5mm,
  colframe=blue!70!black,
  colback=blue!3!white,
  left=3pt,right=3pt,top=2pt,bottom=2pt,
  boxsep=1pt
}

\newtcbox{\ipwmath}{on line,
  boxrule=0.8pt,
  arc=2.5mm,
  colframe=red!70!black,
  colback=red!3!white,
  left=3pt,right=3pt,top=2pt,bottom=2pt,
  boxsep=1pt
}

\newtcbox{\corrmath}{on line,
  boxrule=0.8pt,
  arc=2.5mm,
  colframe=hardgreen,
  colback=green!3!white,
  left=3pt,right=3pt,top=2pt,bottom=2pt,
  boxsep=1pt
}
\usepackage[title,header,toc]{appendix}
\usepackage{minitoc}

\usepackage{import}
\usepackage{xifthen}
\usepackage{pdfpages}
\usepackage{transparent}

\usepackage{fancyhdr}
\fancypagestyle{plain}{
  \fancyhf{}

}
\makeatletter
\def\@makechapterhead#1{%
  \vspace*{25pt}%
  {\parindent \z@ \raggedright \normalfont
    \ifnum \c@secnumdepth >\m@ne
      {\Large\scshape Chapter \thechapter\par\nobreak}
      \vskip 10pt
    \fi
    \interlinepenalty\@M
    {\Huge\bfseries #1\par\nobreak}
    \vskip 12pt
    \hrule height 0.6pt
    \vskip 30pt
  }}

\def\@makeschapterhead#1{%
  \vspace*{25pt}%
  {\parindent \z@ \raggedright \normalfont
    \interlinepenalty\@M
    {\Huge\bfseries #1\par\nobreak}
    \vskip 12pt
    \hrule height 0.6pt
    \vskip 30pt
  }}
\makeatother

\newcommand{\clearrightpage}{%
  \clearpage
  \ifodd\value{page}
  \else
    \thispagestyle{empty}%
    \null
    \clearpage
  \fi
}

\newcommand{\cleartooddpage}{%
  \clearpage
  \ifodd\value{page}
  \else
    \thispagestyle{empty}%
    \null
    \clearpage
  \fi
}

\usepackage{relsize,etoolbox}
\AtBeginEnvironment{quote}{\smaller}

\let\oldquote\quote
\let\endoldquote\endquote
\renewenvironment{quote}[2][]
  {\if\relax\detokenize{#1}\relax
     \def\quoteauthor{#2}%
   \else
     \def\quoteauthor{#2~---~#1}%
   \fi
   \oldquote}
  {\par\nobreak\smallskip\hfill(\quoteauthor)%
   \endoldquote\addvspace{\bigskipamount}}

\newcommand{\partwithoutnumber}[2][]{%
  \begingroup
  \let\oldthispagestyle\thispagestyle
  \renewcommand{\thispagestyle}[1]{\oldthispagestyle{empty}}%
  \part{#2}%
  \if\relax\detokenize{#1}\relax
  \else
    \label{#1}%
  \fi
  \endgroup
}
\usepackage{afterpage}

\usepackage{comment}
\usepackage{here}
\usepackage{makeidx}
\usepackage{pifont}

\usepackage{glossaries}
\makeglossaries

\usepackage{snaptodo}
\snaptodoset{block rise=2em}
\snaptodoset{margin block/.style={font=\tiny}}

\newcommand{\trajrowcolor}[2]{%
  \foreach \x [count=\i from 0] in {#2} {%
    \ifnum\x=1
      \filldraw[#1] (\i*0.32,0) rectangle ++(0.28,0.28);
    \else
      \draw[black] (\i*0.32,0) rectangle ++(0.28,0.28);
    \fi
  }%
}

\renewcommand{\emph}[1]{{\it #1}}

\newcommand{\paren}[2][a]{%
\IfEqCase{#1}{%
{a}{\left(#2\right)}%
{0}{(#2)}%
{1}{\big(#2\big)}%
{2}{\Big(#2\Big)}%
{3}{\bigg(#2\bigg)}%
{4}{\Bigg(#2\Bigg)}%
}[\PackageError{paren}{Undefined option to paren: #1}{}]%
}

\newcommand{\set}[2][a]{
\IfEqCase{#1}{%
{a}{\left\{#2\right\}}%
{0}{\{#2\}}%
{1}{\big\{#2\big\}}%
{2}{\Big\{#2\Big\}}%
{3}{\bigg\{#2\bigg\}}%
{4}{\Bigg\{#2\Bigg\}}%
}[\PackageError{set}{Undefined option to set: #1}{}]%
}

\newcommand{\brac}[2][a]{%
\IfEqCase{#1}{%
{a}{\left[#2\right]}%
{0}{[#2]}%
{1}{\big[#2\big]}%
{2}{\Big[#2\Big]}%
{3}{\bigg[#2\bigg]}%
{4}{\Bigg[#2\Bigg]}%
}[\PackageError{brac}{Undefined option to brac: #1}{}]%
}

\renewcommand{\leq}{\leqslant}
\renewcommand{\geq}{\geqslant}

\makeatletter
\newcommand{\para}{\mathrel{\mathpalette\new@parallel\relax}}
\newcommand{\new@parallel}[2]{%
  \begingroup
  \sbox\z@{$#1T$}
  \resizebox{!}{\ht\z@}{\raisebox{\depth}{$\m@th#1/\mkern-5mu/$}}%
  \endgroup
}
\makeatother

\title{Estimating treatment duration effects via clone-censor-weight: a breast cancer case study}
\date{ }
\author{\small\hfill{} Charlotte \textsc{Voinot}, No\'emie \textsc{Simon-Tillaux}, Emma \textsc{Torrini}, Stefan \textsc{Michiels}, \\ \small Bernard \textsc{Sebastien}, Fabrice \textsc{Andr\'e}, Cl\'ement \textsc{Berenfeld}, and Julie \textsc{Josse}.}

\begin{document}

\maketitle

\section{Introduction}
\subsection{Context}

Randomized clinical trials remain the reference design for evaluating treatment effects, but their eligibility criteria, controlled settings, and limited sample sizes may restrict the transportability of their findings to routine care. Observational data therefore play an important role in evaluating treatment strategies in real-world populations, including strategies that differ by treatment duration \citep{hernan2016,Hernan2020}. In this setting, one may distinguish \emph{static} treatment strategies, which prescribe the same treatment pattern for all individuals, from \emph{dynamic} strategies, which adapt treatment decisions over time according to evolving patient characteristics. For example, a static duration strategy may compare ``stop treatment after 2 years'' versus ``continue treatment for 5 years,'' whereas a dynamic strategy may recommend stopping treatment only when a biomarker, toxicity profile, or disease status crosses a pre-specified threshold during follow-up \citep{Murphy2003,Cain_2010}. 

Evaluating such longitudinal strategies in observational data is challenging because treatment and covariate histories evolve jointly over time. Post-baseline covariates may both affect subsequent treatment decisions and predict the survival outcome, while themselves being affected by prior treatment, thereby inducing time-varying confounding \citep{Robins2000MSM,Hernan2000-hiv}. When the outcome is time-to-event, an additional difficulty arises: individuals with longer observed treatment durations must, by construction, survive and remain event-free long enough to receive them, which may induce immortal time bias under naive comparisons \citep{Suissa_2008}. In this article, we study static treatment duration strategies under two progressively more complex settings: first, when confounding is fully captured by baseline covariates, and second, when it also involves time-varying covariates measured during follow-up. In both settings, we rely on the cloning--censoring--weighting (CCW) framework \citep{Hernan_2006,Cain_2010}, which emulates a target trial by replicating each individual across the strategies under comparison, artificially censoring clones when their observed treatment history becomes incompatible with the assigned strategy, and correcting the resulting informative censoring through inverse probability of censoring weighting.

\subsection{Motivations}
Our motivating application concerns adjuvant endocrine therapy in early-stage breast cancer, and in particular tamoxifen, which is widely used in hormone receptor-positive disease. Historical randomized trials and large patient-level meta-analyses have established that 5 years of adjuvant tamoxifen reduces both recurrence and breast cancer mortality compared with no endocrine treatment, supporting its role as a standard of care \citep{tamoxifen_5y,EBCTCG2011}. Although randomized studies have also examined whether extending therapy beyond 5 years may provide additional benefit in selected patients \citep{Davies2013,Gray2013}, treatment burden remains substantial in practice.

Indeed, tamoxifen is associated with adverse effects and reduced quality of life, which often compromise long-term adherence. In routine care, observational studies have reported substantial early discontinuation, with nearly 55\% of patients stopping treatment before completing the recommended 5-year duration \citep{Van_Herk-Sukel2010-vo,Hershman2010}. This raises a clinically important question: could shorter treatment durations provide similar survival benefit for some patients while reducing treatment burden? Although randomized trials have explored shorter durations, such as 2 versus 5 years of tamoxifen, the available evidence remains limited and does not fully reflect the heterogeneity, adherence patterns, and treatment pathways observed in routine care \citep{Delozier2000-ba}.

To investigate this question in a real-world setting, we use data from a \emph{Breast Cancer} cohort, a large French nationwide prospective study of patients with localized breast cancer. Its detailed longitudinal information on treatment exposure, follow-up, and patient characteristics makes it well suited for evaluating alternative treatment duration strategies. In this work, we emulate a target trial comparing static tamoxifen duration strategies using the cloning--censoring--weighting framework.
\subsection{Methodological challenges, contributions and article outline}

The cloning--censoring--weighting (CCW) framework is increasingly used to emulate treatment duration strategies in observational survival studies, yet several methodological issues remain insufficiently clarified in the current literature \citep{Maringe2020,wang2022statistical,Dumas_2025,Amiot_2024,Petito_2020}. More broadly, CCW belongs to the wider class of longitudinal causal inference methods for time-varying treatments and survival outcomes, alongside the parametric G-formula, 
marginal structural models estimated by inverse probability weighting, and structural nested models estimated by G-estimation \citep{robins1986new,Robins2000MSM,Hernan2000-hiv,Vansteelandt2014,Keogh2023SequentialTrials}. We focus on CCW because it provides a particularly transparent operational framework for treatment duration strategies in observational survival data, especially when the aim is to account for observed treatment histories that provide admissible approximations of the target strategy, avoid immortal time bias, and compare clinically meaningful static duration rules.

This work addresses three main challenges. First, clinically realistic duration strategies often rely on relaxed intervention rules, such as grace periods, which require assumptions linking admissible observed trajectories to the target treatment strategy. In Section~\ref{sec-notations-chp4}, we formalize the corresponding causal estimands, clarify the assumptions required for cloning and admissibility, and illustrate how immortal time bias arises under naive comparisons based on observed treatment duration. Second, CCW induces artificial censoring by design, which may coexist with natural censoring. In Sections~\ref{sec-baseline} and~\ref{sec-tdconf}, we distinguish these two mechanisms, clarify when separate censoring weights can be validly combined, and study the impact of ignoring natural censoring or misspecifying censoring models. Third, IPCW may suffer from unstable weights and finite-sample variability, while alternative approaches remain little explored in this setting. We therefore compare, after cloning and artificial censoring, IPCW with outcome regression and augmented estimators under both baseline and time-varying confounding, and provide practical guidance regarding robustness, computational burden, and sensitivity to model misspecification. Finally, in Section~\ref{sec-application}, we illustrate the proposed methodology by emulating alternative tamoxifen duration strategies using data from the Breast Cancer cohort.

\section{Causal framework for treatment duration: estimands, notation, and identification challenges}\label{sec-notations-chp4}

\subsection{Notation and causal estimand: static duration strategies}

We consider a longitudinal observational setting with scheduled follow-up times indexed by $k = 0,1,\dots,K$, where treatment decisions and covariates are recorded at each time point. We assume that time $k$ corresponds to visit $k$, and at each time point we observe a binary treatment indicator
\[
A_k \in \{0,1\},
\]
as well as baseline and possibly time-dependent covariates, denoted
collectively by $X_k$. We write
\[
\bar A_k = (A_0,\dots,A_k), \qquad \bar X_k = (X_0,\dots,X_k),
\]
for the observed treatment and covariate histories up to visit $k$.
Treatment is initiated at baseline, so that $A_0 = 1$ for all individuals.

Individuals are followed from baseline until the occurrence of the event
of interest or censoring, whichever occurs first. Let $T$ denote the event time, $C$ the censoring time, $\Tilde{T}=\min(T,C)$ the observed time and $\Delta=\mathbbm{1}\{T\leq C\}$ the event status. 

The assumed data structure is a discrete-time formulation, where the event indicator corresponds to whether the event occurs between two consecutive visits. One can view this representation as a discretization of an underlying continuous-time process: between each pair of visits, the follow-up could be partitioned into a sequence of finer sub-intervals in which events (and censoring) may occur, while treatment and covariates are assumed constant between visits (as in \cite{Keogh2023SequentialTrials}). As the sub-intervals become arbitrarily small, this discrete-time description approaches the corresponding continuous-time setting.
Treatment and covariates, which are only observed at visit times, are assumed constant between visits and are therefore carried forward across the intervals
contained between two consecutive visits. 

\paragraph{Static duration strategies.} 
Treatment duration is often summarized retrospectively from the observed
treatment history $\bar A_K$. However, such summaries generally do not
correspond to well-defined interventions and therefore do not define
causal estimands. Instead, causal effects of treatment duration must be
defined in terms of sustained strategies specified at baseline. 

For a given duration $d \in \{1,\dots,K+1\}$, we define the associated
\emph{static duration strategy} $g_d$ as the baseline-assigned
hypothetical intervention 
\[
g_d(k) =
\begin{cases}
1, & k < d,\\
0, & k \ge d,
\end{cases}
\qquad k = 0,\dots,K.
\]

\paragraph{Causal estimand.}
Let $T^{(d)}$ denote the counterfactual event time that would occur under the baseline-defined strategy $g_d$, and define the corresponding survival function
\[
S^{(d)}(t) = \mathbb P(T^{(d)} > t), \qquad t \ge 0.
\]
For two distinct static duration strategies $d_0,d_1 \in \{1,\dots,K+1\}$, the causal restricted mean survival time (RMST) contrast is defined as
\[
\theta_{\mathrm{RMST}}(d_1,d_0)
=
\mathbb E\!\left[
\min(T^{(d_1)},\tau)
-
\min(T^{(d_0)},\tau)
\right]
=
\int_0^\tau
\left[
S^{(d_1)}(t)
-
S^{(d_0)}(t)
\right] \, \mathrm{d}t,
\]
where $\tau>0$ denotes the time horizon, and the duration strategies $d_0$ and $d_1$ satisfy $d_0, d_1 \le \tau$.

In the next section, we show that conditioning the outcome on observed treatment durations does not identify this causal RMST ($\mathbb{E}[T^{(d)}]$), even in the absence of time-varying confounding or informative censoring.

\subsection{Immortal time bias in observational settings}

A naive approach to evaluating treatment duration effects is to compare survival outcomes according to realized treatment durations. Let $D$ denote the observed treatment duration until discontinuation,
\[
D \equiv D(\bar A_K)
=
\max\{d\leq T \wedge C, A_d=1\}.
\]
Stratifying on $D=d$ amounts to comparing
$\mathbb E[\min(T,\tau)\mid D=d]$. However, in the absence of censoring, the event $\{D=d\}$ can occur only among individuals who remain event-free until time $d$, i.e., $\{T>d\}$. Moreover, among individuals whose observed treatment history matches $g_d$ up to $d$, consistency implies $\{T>d\}\equiv\{T^{(d)}>d\}$. Consequently,
\[
\mathbb E[\min(T,\tau)\mid D=d]
=
\mathbb E\!\left[\min(T^{(d)},\tau)\mid T^{(d)}>d\right],
\]
which differs from the marginal estimand
$\mathbb E[\min(T^{(d)},\tau)]$. Therefore, in general, 
\[
\mathbb E[\min(T,\tau)\mid D=d]
\neq
\mathbb E[\min(T^{(d)},\tau)]
\qquad
(\text{and likewise } \mathbb E[T\mid D=d]\neq \mathbb E[T^{(d)}]).
\]

This discrepancy follows from the law of total expectation:
\[
\mathbb E[\min(T^{(d)},\tau)]
=
\mathbb E\!\left[\min(T^{(d)},\tau)\mid T^{(d)}>d\right]\mathbb P(T^{(d)}>d)
+
\mathbb E\!\left[\min(T^{(d)},\tau)\mid T^{(d)}\le d\right]\mathbb P(T^{(d)}\le d).
\]
As soon as $\mathbb P(T^{(d)}\le d)>0$, the marginal mean incorporates early events before $d$, whereas the conditional quantity excludes them. This discrepancy is precisely the \textit{immortal time bias} induced by conditioning on realized duration.

To illustrate this, consider two studies differing only in the treatment assignment mechanism (Figure~\ref{fig:immortal_time}). In a randomized trial, individuals are assigned at baseline to receive either 2 or 5 years of treatment and compared according to the intention-to-treat principle\footnote{Outcomes are compared between groups defined by baseline randomization, irrespective of treatment discontinuation or post-baseline changes}. By construction, the causal contrast is zero in this illustrative example.

\begin{figure}[H]
    \centering
    \includegraphics[width=0.8\linewidth]{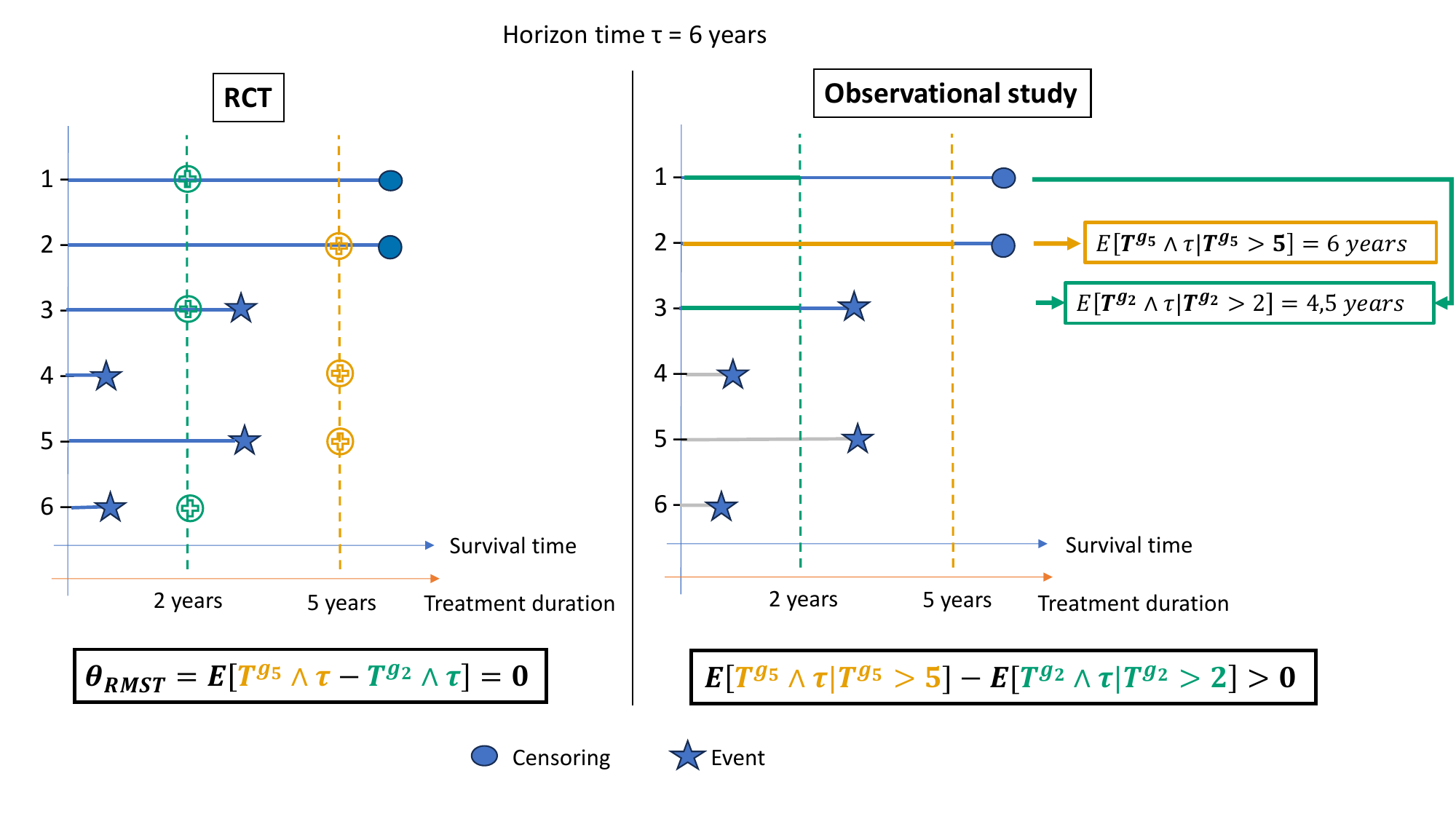}
    \caption{Difference in RMST between 2 and 5 years of treatment: illustration of immortal time bias in an observational study. In the randomized controlled trial (left panel), individuals are assigned at baseline to a treatment strategy of either 5 years (\textcolor{orange_cbf1}{orange circle}) or 2 years (\textcolor{green_cbf1}{green circle}). In the observational study (right panel), treatment duration is defined retrospectively: a continuous \textcolor{orange_cbf1}{orange line} denotes an observed treatment duration of 5 years, a continuous \textcolor{green_cbf1}{green line} denotes an observed duration of 2 years, and a \textcolor{gray}{gray line} denotes individuals whose observed treatment duration does not correspond to either strategy.}
    \label{fig:immortal_time}
\end{figure}

In contrast, in an observational study, comparisons are based on realized treatment duration in a per-protocol-type\footnote{Individuals are compared according to their observed adherence to the treatment strategy, rather than their initial assignment} analysis. Individuals~1 and~3 complete 2 years of treatment, whereas individual~2 completes 5 years. Conditioning on realized duration excludes individuals who fail or discontinue treatment before reaching the target duration, which induces a spurious association between longer treatment duration and improved survival. As a result, the empirical RMST difference favors longer treatment despite a true causal contrast of zero.

This bias does not arise from the estimand itself, but from conditioning on a post-baseline quantity defined only among survivors. A valid comparison therefore requires defining treatment duration as a baseline intervention specifying how long treatment would be maintained under a hypothetical strategy. In observational settings, such strategies are not directly observed and must be emulated from the available data.

In practice, treatment discontinuation is rarely observed exactly at prespecified times in observational studies. Observed treatment durations are therefore heterogeneous and often do not align with the exact decision times defining a baseline strategy $g_d$. Consequently, assessing adherence to $g_d$ can be ambiguous when treatment changes are observed only at irregular or patient-specific visit times. This motivates a slight relaxation of the target strategy when defining adherence in the emulation procedure, which we formalize next.

\subsection{Causal identification under post-baseline treatment decisions}

We consider two settings to study the identification of static duration strategies.

In Scenario~(1), treatment continuation at each decision time depends only on baseline confounders. In Scenario~(2), it may additionally depend on time-varying confounders measured during follow-up. In both settings, we distinguish two forms of \emph{natural censoring}: (a) independent censoring and (b) informative censoring handled under a conditionally independent censoring assumption. Our goal is to identify the causal RMST contrast $\theta_{\mathrm{RMST}}$ and clarify how the required assumptions change across these settings.

Let $T$ and $C$ denote the event and censoring times, and let $R_k=\mathbbm{1}\{\min(T,C)>k\}$ denote the at-risk indicator at time $k$. In the time-varying setting, we additionally introduce $K_k$, the indicator of natural censoring at time $k$, and $Y_k$, the indicator of the event at time $k$. We write $(X, X_k)$ for baseline and time-varying confounders, and $(X_c, X_{c,k})$ for censoring-related covariates. For simplicity, confounding and censoring-related variables are assumed disjoint.

Table~\ref{tab:assumptions_baseline} and Table~\ref{tab:assumptions_timedep} summarize the identification assumptions. In both scenarios, admissible consistency links the observed data to the relevant potential quantities under treatment histories compatible with the target strategy. In particular, any observed treatment trajectory in $\mathcal A_d$ is assumed to yield the same potential outcome as $g_d$. This strong assumption implies that deviations allowed in the adherence definition do not affect the outcome. In Scenario~(1), consistency concerns the outcome and censoring processes, whereas in Scenario~(2) it also extends to the time-varying covariate process. When the admissible set reduces to the single target strategy, 
this corresponds to the usual consistency assumption of the potential outcomes framework \citep{rubin_1974}. 

Treatment exchangeability requires treatment assignment at each decision time $k$ to be independent of $T^{(d)}$ conditional on the observed treatment and covariate history. Positivity requires both treatment options to remain possible within each such history with positive probability.

For censoring, case (a) assumes independence between $C$ and $T^{(d)}$, whereas case (b) allows censoring to depend on observed covariates and instead relies on conditional independent censoring given the observed history, together with the corresponding positivity assumptions.

\begin{table}[H]
\centering
\renewcommand{\arraystretch}{1.3}

\begin{tabular}{p{6cm} p{8.2cm}}
\hline
\textbf{Assump.} 
& \textbf{Scenario (1): Baseline covariates} \\
\hline

\scriptsize{\textbf{Admissible consistency (treatment)}} &
$\assumplabelstep{ass:adm-treat-consistency}{Admissible consistency for treatment}\bar A_k \in \mathcal A_{d,k} \;\Rightarrow\; \mathbb{I}\{T \leq k\} = \mathbb{I}\{T^{(d)} \leq k\}$ \\
\hline

\scriptsize{\textbf{Admissible consistency (censoring)}} &
$\assumplabelstep{ass:adm-cens-consistency}{Admissible consistency for censoring}\bar A_k \in \mathcal A_{d,k} \;\Rightarrow\; \mathbb{I}\{C \leq k\} = \mathbb{I}\{C^{(d)} \leq k\}$ \\
\hline

\scriptsize{\textbf{Exchangeability (treatment)}} &
$\assumplabelstep{ass:texch-baseline}{Treatment exchangeability (baseline confounders)}T^{(d)}\;\perp\!\!\!\perp\; \bar A \mid X$ \\
\hline

\scriptsize{\textbf{Exchangeability (censoring)}} &
$\assumplabelstep{ass:texch-baseline-cens}{Treatment exchangeability (baseline confounders)}C^{(d)}\;\perp\!\!\!\perp\; \bar A \mid X$ \\
\hline

\scriptsize{\textbf{Positivity (treatment)}} &
$\assumplabelstep{ass:tpos-treat-baseline}{Treatment positivity (baseline confounders)}\mathbb P(\bar A \in \mathcal A_d \mid X) > 0$ \\
\hline

\scriptsize{\textbf{Positivity (censoring) (a)}} &
$\assumplabelstep{ass:tpos-cens-baseline}{Treatment positivity (baseline confounders)}\mathbb P(C^{(d)} \geq k) > 0$ \\
\hline

\scriptsize{\textbf{Positivity (censoring) (b)}} &
$\assumplabelstep{ass:tpos-cond-baseline}{Treatment positivity (baseline confounders)}\mathbb P(C^{(d)} \geq k \mid X_c) > 0$ \\
\hline

\scriptsize{\textbf{Independent censoring (a)}} &
$\assumplabelstep{ass:cens-indep}{Independent natural censoring}C^{(d)} \;\perp\!\!\!\perp\; T^{(d)}$ \\
\hline

\scriptsize{\textbf{Conditionally independent censoring (b)}} &
$\assumplabelstep{ass:cens-cond-baseline}{Conditionally independent natural censoring (baseline confounders)}C^{(d)}\;\perp\!\!\!\perp\; T^{(d)} \mid X_c$ \\
\hline

\end{tabular}
\caption{Causal assumptions for the scenario with baseline covariates (1) under independent (a) or conditionally independent censoring (b). Conditions involving $k$ hold for all value of $k \in \{0,\dots,K-1\}$.}
\label{tab:assumptions_baseline}
\end{table}

\begin{table}[H]
\centering
\renewcommand{\arraystretch}{1.3}

\begin{tabular}{p{6cm} p{8.5cm}}
\hline
\textbf{Assump.} 
& \textbf{Scenario (2): Time-dependent covariates} \\
\hline

\scriptsize{\textbf{Admissible consistency (treatment)}} &
$\assumplabelstep{ass:sutva-treat-timedep}{Admissible consistency for treatment (time-dependent confounders)}
\bar A_k \in \mathcal A_{d,k} \;\Rightarrow\;  Y_{k+1}=Y_{k+1}^{(d)}$
\\

\hline

\scriptsize{\textbf{Admissible consistency (censoring)}} &
$\assumplabelstep{ass:sutva-cens-timedep}{Admissibile consistency for censoring (time-dependent confounders)}
\bar A_k \in \mathcal A_{d,k} \;\Rightarrow\; K_{k+1} = K_{k+1}^{(d)}$
\\

\hline

\scriptsize{\textbf{Admissible consistency (covariates)}} &
$\assumplabelstep{ass:sutva-cov-timedep}{Admissible consistency for covariates (time-dependent confounders)}
\bar A_k \in \mathcal A_{d,k} \;\Rightarrow\; X_{k+1}=X_{k+1}^{(d)}$
\\

\hline

\scriptsize{\textbf{Exchangeability (treatment)}} &
$\assumplabelstep{ass:texch-timedep}{Treatment exchangeability (time-dependent confounders)}
\forall \ell \geq k+1, Y^{(d)}_\ell, X_{\ell}^{(d)} \;\perp\!\!\!\perp\; A_k \mid \bar A_{k-1}, \bar X_{k}, C_{k}=Y_{k}=0$
\\

\hline



\scriptsize{\textbf{Positivity (treatment)}} &
$\assumplabelstep{ass:tpos-timedep}{Treatment positivity (time-dependent confounders)}
\mathbb P(\bar A_k \in \mathcal A_{d,k} \mid \bar X_{k}, K_{k}=Y_{k}=0) > 0$
\\

\hline

\scriptsize{\textbf{Positivity (censoring)}} &
$\assumplabelstep{ass:cens-cond-pos-timedep}{Conditional censoring positivity (time-dependent confounders)}
\exists \,\epsilon > 0, \forall\, k\in[K],\, \mathbb P(K_k^{(d)}=0\mid \bar X_{c,k-1}^{(d)}, Y^{(d)}_{k-1}= K_{k-1}^{(d)}=0) > \epsilon$
\\

\hline

\scriptsize{\textbf{Conditionally independent censoring}} &
$\assumplabelstep{ass:cens-cond-timedep}{Conditionally independent natural censoring (time-dependent confounders)}
\forall \ell \geq k+1, Y_\ell^{(d)}, X_\ell^{(d)} \;\perp\!\!\!\perp\; K_{k+1}^{(d)} \mid \bar A_k, \bar X_{c,k}, \bar X_{k}, K_{k}=Y_{k}=0$
\\

\hline

\end{tabular}

\vspace{0.2cm}

\caption{Causal assumptions for the scenario with time-dependent covariates (2) under conditionally independent censoring (b). All assumptions are understood to hold for all decision times $k=1,\dots,K$ and all strategies $g_d$ of interest.}
\label{tab:assumptions_timedep}

\end{table}

We now turn to the estimation of the causal effect of treatment duration. We first consider a setting with baseline confounding only, which provides a simplified framework to introduce the proposed methods before addressing time-varying confounding.

\section{Causal effect of treatment duration in the presence of baseline confounders}\label{sec-baseline}

This setting is used as a pedagogical starting point to introduce the methodology in a simplified framework. In practice, however, the main sources of bias typically arise from time-varying confounding, which we address in the next section. Still, although our motivation here is primarily pedagogical, some applied studies have used baseline covariates alone to evaluate post-baseline treatment strategies \citep{Maringe2020}.
The scenario (1) is summarized by the DAG Figure \ref{fig:dag_discrete}.

\begin{figure}[H]
\centering
\begin{tikzpicture}[->, >=stealth, node distance=1.8cm and 1.8cm, thick]
    \def\r{0.9cm}
    \node (X)  [draw, circle, minimum size=\r] at (3,1) {$X$};
    \node (A)  [draw, circle, minimum size=\r] at (0,0) {$\bar A$};
    \node (T)  [draw, circle, minimum size=\r] at (6,0) {$T$};
    \node (C)  [draw, circle, minimum size=\r] at (6,2) {$C$};
    \node (Tobs) [draw, circle, minimum size=\r] at (8,0) {$\tilde T$};
    \node (Xc) [draw, circle, minimum size=\r] at (5,1) {$X_c$};
    \node (D) [draw, circle, minimum size=\r] at (6,-1) {$D$};

    \draw[red_cbf1] (X) -- (A);
    \draw[red_cbf1] (X) -- (T);

    \draw (A) -- (T);
    \draw[navy_cbf1] [bend left=20]  (A) to (C);
    \draw[purple_cbf1] (A) -- (D);
    \draw[purple_cbf1] (Tobs) -- (D);
    \draw[navy_cbf1] (Xc) -- (C);
    \draw[navy_cbf1] (Xc) -- (T);

    \draw (T) -- (Tobs);
    \draw (C) -- (Tobs);

\end{tikzpicture}
\caption{Simplified causal diagram for the baseline-confounding setting. Here, $X$ denotes baseline confounders of treatment and the event time, $\bar A$ the treatment history, $D$ the observed duration, $X_c$ baseline covariates related to natural censoring, $T$ the event time, $C$ the natural censoring time, and $\tilde T=\min(T,C)$ the observed follow-up time. \textcolor{red_cbf1}{Red edges} represent pathways linking baseline confounders to both treatment and outcome. \textcolor{navy_cbf1}{Blue edges} represent pathways involved in the censoring mechanism. \textcolor{purple_cbf1}{Purple edges} highlight the treatment-duration structure that may induce selection mechanisms when restricting to admissible trajectories.}
\label{fig:dag_discrete}
\end{figure}
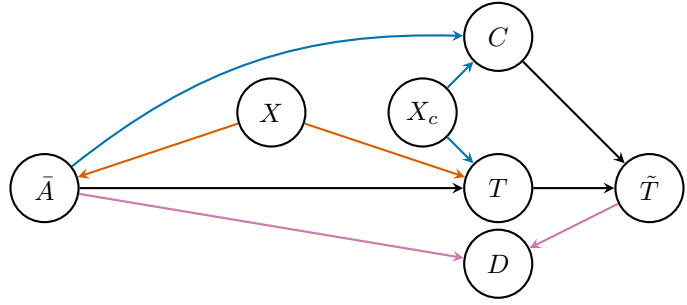

Figure~\ref{fig:dag_discrete} represents a discrete-time causal structure in which treatment is initiated at baseline and subsequently continued or discontinued over time. Baseline covariates $X$ influence both treatment decisions and the event process, while no time-varying confounders of the treatment--outcome relationship are present. Despite this simple structure, treatment duration remains a post-baseline quantity realized only among individuals who survive and remain uncensored long enough.

\subsection{Cloning-Censoring}\label{sec-cloningcensoring}

\paragraph{Cloning.}


We consider the comparison of two static treatment duration strategies $d_0$ and $d_1$. For each individual $i=1,\dots,n$, we observe
\[
(\widetilde T_i,\Delta_i, X_i, X_{c,i}, \{A_{i,k}\}_{k=0}^{\widetilde T_i}),
\]
where $X_i$ denotes baseline confounders, $X_{c,i}$ baseline covariates associated with natural censoring, and $A_{i,k}$ the treatment indicator at decision time $k$. The treatment history is observed up to the follow-up time $\widetilde T_i=\min(T_i,C_i)$, with event indicator $\Delta_i= \mathbb{I}\{T_i < C_i\}$.


To emulate the corresponding two-arm trial, each individual is cloned into two copies assigned to $g_{d_0}$ or $g_{d_1}$. Both clones share the same observed processes
$(T_i, C_i, X_i, X_{c,i}, \{A_{i,k}\}_{k=0}^K)$
and differ only by their assigned strategy. 

Each clone is artificially censored at the first deviation from its assigned strategy, yielding clone-specific observed outcomes.

\paragraph{Artificial censoring.}

For a clone $(i,d)$ assigned to strategy $g_{d}$, artificial censoring is
introduced when deviating to the assigned strategy.
We define the discrete-time indicator of deviation at decision time $k$ as 

\[
\Gamma_{i,k}^{(d)}
=
\mathbbm{1}\{\bar A_{i,k}\notin \mathcal A_{d,k}\},
\qquad k=0,\dots,K,
\]
By convention, $\Gamma_{i,0}^{(d)}=0$ for all $i$ and $d$, so that all
clones are compliant with their assigned strategy at baseline.
The artificial censoring time is
\[
G_i^{(d)}=\inf\{\, k : \Gamma_{i,k}^{(d)}=1 \,\},
\]
with the convention $G_i^{(d)}=\tau$ if $\Gamma_{i,k}^{(d)}=0$ for all
$k=0,\dots,K$. The observed and the corresponding truncated follow-up time and the event indicator for clone $(i,d)$ is
\[
\widetilde T_i^{(d)}=\min\bigl(T_i,\, C_i,\, G_i^{(d)}\bigr), \qquad \Delta_i^{(d)}
=
\mathbbm{1}\{T_i \le C_i,\ T_i \le G_i^{(d)}\}.
\]

To clarify the cloning-censoring procedure, Figure~\ref{fig:cloning} considers two treatment strategies (\textcolor{green_cbf1}{2 years} versus \textcolor{purple_cbf1}{5 years} of treatment) and illustrates how individuals are cloned and then artificially censored under each strategy.

\begin{figure}[H]
    \centering
    \includegraphics[width=0.9\linewidth]{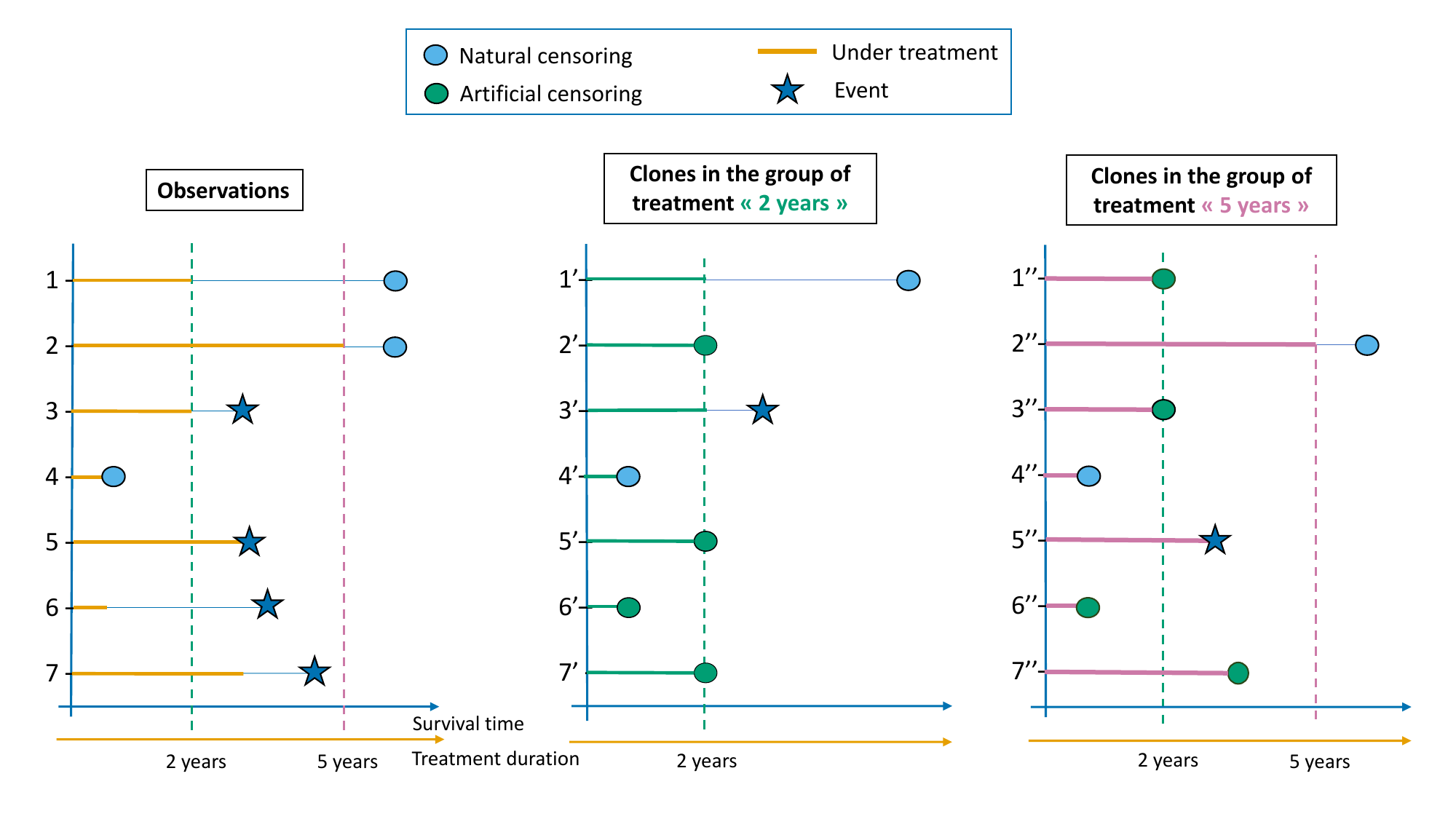}
    \caption{Cloning–-censoring procedure for comparing \textcolor{green_cbf1}{2-year} versus \textcolor{purple_cbf1}{5-year} treatment strategies}
    \label{fig:cloning}
\end{figure}
Because cloning is deterministic and performed independently of the observed data, the assigned strategy $g_d$ is independent of baseline covariates and potential outcomes. Figure~\ref{fig:dag_cloning} represents the causal structure induced by cloning and artificial censoring.


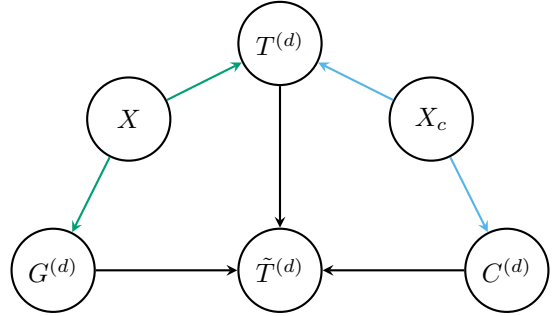
\begin{figure}[H]
    \centering
\begin{tikzpicture}[->, >=stealth, node distance=1.6cm and 1.5cm, thick]
    \def\r{1.1cm}
    \node (X) [draw, circle,minimum size=\r] at (1,2) {$X$};
    
    \node (G) [draw, circle,minimum size=\r] at (0,0) {$G^{(d)}$};


    \node (Tobs) [draw, circle, minimum size=\r] at (3,0) {$\tilde T^{(d)}$};
    \node (T) [draw, circle, minimum size=\r] at (3,3) {$T^{(d)}$};

    \node (C) [draw, circle,minimum size=\r] at (6,0) {$C^{(d)}$};

    \node (Xc) [draw, circle,minimum size=\r] at (5,2) {$X_c$};
    
    \draw[green_cbf1] (X) -- (G);
    \draw[green_cbf1] (X) -- (T);

    \draw[blue_cbf1] (Xc) -- (C);
    \draw[blue_cbf1] (Xc) -- (T);
    
    \draw (T) -- (Tobs);
    \draw (C) -- (Tobs);
    \draw (G) -- (Tobs);









\end{tikzpicture}
\caption{
Discrete-time representation of treatment assignment, event, and censoring processes after cloning--censoring. $T^{(d)}$ denotes the potential event time, $\tilde T^{(d)}=\min(T^{(d)},C^{(d)},G^{(d)})$ the observed follow-up time under cloning--censoring, $G^{(d)}$ the artificial censoring time, and $C^{(d)}$ the natural censoring time under strategy $g_d$. Here, $X$ denotes baseline confounders and $X_c$ censoring-related covariates. \textcolor{blue_cbf1}{Blue edges} represent informative natural censoring, and \textcolor{green_cbf1}{green edges} informative artificial censoring.}


\label{fig:dag_cloning}
\end{figure}

After cloning and artificial censoring, the resulting data structure can be represented by a DAG mimicking a randomized trial, in the sense that baseline covariates $X$ no longer confound strategy assignment. However, $X$ may still predict deviations from the assigned strategy and therefore drive artificial censoring. Because this censoring is informative, identifying the causal quantity requires additional assumptions on the censoring processes in the cloned dataset.

Importantly, artificial censoring is strategy-specific, since deviations are defined relative to the assigned strategy $g_d$. In contrast, natural censoring originates from the same underlying data-generating process, although its role may depend on the comparison considered.


\begin{proposition}[Sequential censoring exchangeability and positivity on the cloned risk set]
\label{prp:seq-cens-riskset}
Consider the cloned and artificially censored dataset. For each individual $(i,d)$, assigned to strategy $g_d$, we assume that the following
conditions hold among individuals.

\begin{enumerate}
\item[(i)] \emph{Conditionally independent artificial censoring.}
Under Assumptions~\ref{ass:texch-baseline} (\nameref{ass:texch-baseline}) and ~\ref{ass:adm-treat-consistency} (\nameref{ass:adm-treat-consistency}), we assume that
\[
T^{(d)} \;\perp\!\!\!\perp\; G^{(d)} \ \Big|\ X.
\]

\item[(ii)] \emph{Natural censoring exchangeability (independent or conditional).}
Either, under Assumptions~\ref{ass:cens-indep} (\nameref{ass:cens-indep}) and ~\ref{ass:adm-cens-consistency} (\nameref{ass:adm-cens-consistency}), we assume that 
\[
T^{(d)} \;\perp\!\!\!\perp\; C,
\]
or, under Assumptions~\ref{ass:cens-cond-baseline} (\nameref{ass:cens-cond-baseline}) and ~\ref{ass:adm-cens-consistency} (\nameref{ass:adm-cens-consistency}), we assume that
\[
T^{(d)} \;\perp\!\!\!\perp\; C^{(d)}
\ \Big|\ X_{c}.
\]

\item[(iii)] \emph{Positivity of the censoring process.}
There exists $\epsilon>0$ such that, for each $d$,
\[
\mathbb P\!\left(G^{(d)} \ge t\mid X\right)\ge \epsilon,
\]
and, for natural censoring, either under Assumptions~\ref{ass:cens-indep} (\nameref{ass:cens-indep}) and \ref{ass:tpos-cens-baseline} (\nameref{ass:tpos-cens-baseline}),
\[
\mathbb P(C \ge t)\ge \epsilon,
\qquad t\in[0,\tau),
\]
or under Assumptions~\ref{ass:cens-cond-baseline} (\nameref{ass:cens-cond-baseline}) and
\ref{ass:tpos-cond-baseline} (\nameref{ass:tpos-cond-baseline}),
\[
\mathbb P\!\left(C^{(d)} \ge t \mid X_{c}\right)\ge \epsilon, \qquad t\in[0,\tau).
\]
\end{enumerate}
\end{proposition}


The conditions in Proposition~\ref{prp:seq-cens-riskset} are not directly testable, as they involve the counterfactual event time $T^{(d)}$. However, they can be partially assessed by verifying that, within strata of baseline covariates, individuals have non-negligible probabilities of (i) remaining compatible with each assigned strategy and (ii) remaining under follow-up over time. Importantly, the positivity requirement in (iii) is a support condition: within each relevant covariate stratum and follow-up period, individuals must be able to both remain compatible with the assigned strategy and remain under follow-up. Otherwise, deterministic deviations or losses to follow-up may lead cloned risk sets to become empty or nearly empty for one strategy at specific decision times, indicating limited support for emulating the target strategy.
One of the most widely used approaches to address informative censoring is inverse probability of censoring weighting.


In the context of treatment duration
strategies, this approach is commonly implemented within the
cloning--censoring--weighting framework.



\subsection{Inverse Probability of Censoring Weighting (IPCW)}\label{sec-baselineipcw}

In the following subsection, we focus on the setting in which selection bias arises from artificial censoring only or in combination with natural censoring.

Under Proposition~\ref{prp:seq-cens-riskset} (\nameref{prp:seq-cens-riskset}) and Assumptions~\ref{ass:cens-indep}
(\nameref{ass:cens-indep}), selection bias arises only from artificial censoring. For each strategy $g_d$, define the (unstabilized) artificial-censoring weight as
\begin{equation}
\label{eq:ipcw-art}
W_i^{(d)}(0)=1,
\qquad
W_i^{(d)}(t)
=
\frac{1}{\mathbb P\!\left(G_i^{(d)} \ge t \mid X_i\right)}, t \in[0,\tau).
\end{equation}
That is, $W_i^{(d)}(t)$ is the inverse probability that the observed treatment history remains within the admissible set associated with the assigned strategy up to time $t$, conditional on baseline covariates.

Under Proposition~\ref{prp:seq-cens-riskset} (\nameref{prp:seq-cens-riskset}) and Assumptions~\ref{ass:cens-cond-baseline} (\nameref{ass:cens-cond-baseline}) and \ref{ass:tpos-cond-baseline} (\nameref{ass:tpos-cond-baseline}), clones may be subject to both informative artificial censoring and natural censoring. Identification under strategy $g_d$ is obtained by weighting each clone by the inverse probability of remaining observable under both mechanisms. The (unstabilized) global censoring weight is defined as 

\begin{equation}
\label{eq:ipcw-tot}
W_i^{(d)}(0)=1,
\qquad
W_{i}^{(d)}(t)
=
\frac{1}{
\mathbb P\!\left(G^{(d)} \ge t, C^{(d)} \ge t \mid X_{c,i},X_i\right)
}, t \in[0,\tau).
\end{equation}

By the chain rule 
$$
\mathbb P\!\left(G^{(d)} \ge t, C^{(d)} \ge t \mid X_{c,i},X_i\right) = \mathbb P\!\left(G^{(d)} \ge t, \mid X_{c,i},X_i\right) \times \mathbb P\!\left(C^{(d)} \ge t, \mid X_{c,i},X_i, G^{(d)} \ge t \right)
$$
which does not require independence between artificial and natural censoring.

Under the additional assumption that the two sets of baseline covariates are disjoint and that, conditionally on $(X_i,X_{c,i})$, artificial and natural censoring are independent, the overall censoring mechanism factorizes as
\[
\mathbb P\!\left(G^{(d)} \ge t, C^{(d)}\ge t \mid X_{c,i},X_i\right)
=
\mathbb P\!\left(G^{(d)} \ge t \mid X_i\right)\,
\mathbb P\!\left(C^{(d)}\ge t \mid X_{c,i}\right).
\]

These assumptions are not required for identification, but may simplify modeling and estimation of the censoring mechanisms.

Importantly, artificial censoring is strategy-specific and must therefore be modeled separately within each treatment arm. In contrast, natural censoring arises from the same underlying data-generating process, although whether it should be modeled jointly or separately across arms depends on the assumptions imposed.

\begin{definition}[IPCW-Kaplan-Meier for censoring]\protect\hypertarget{def-ipcwkm-chp4}{}\label{def-ipcwkm-chp4}
We let $\widehat W_i^{(d)}(\cdot)$ denote an estimator of the individual censoring weight defined in either equation~\ref{eq:ipcw-art} or \ref{eq:ipcw-tot} for strategy $g_{d}$. We introduce
\begin{align*}
D_j^{\mathrm{IPCW}}(d) &:= \sum_{i=1}^n \widehat W_i^{(d)}(t_j) \Delta_{i}^{(d)} \mathbb{I}(\widetilde T_i^{(d)} = t_j) \\
\quad\text{and}\quad N^{\mathrm{IPCW}}_j(d) &:= \sum_{i=1}^n \widehat W_i^{(d)}(t_j) \mathbb{I}(\widetilde T_i^{(d)} \geqslant t_j),
\end{align*}
be the weighted numbers of deaths and individuals at risk at time \(t_j\). The IPCW version of the KM estimator is defined as:
\[
\widehat{S}^{(d)}_{\mathrm{IPCW}}(t)
=
\prod_{t_j \leqslant t}\left(1-\frac{D_j^{\mathrm{IPCW}}(d)}{N_j^{\mathrm{IPCW}}(d)}\right). 
\]
\end{definition}

The weighted Kaplan--Meier estimator redistributes the contribution of censored individuals to similar individuals who remain at risk. 
Individuals more likely to be censored, given their covariates, receive larger weights when they remain under follow-up. The RMST estimator is then
\begin{equation}\phantomsection\label{eq-thetaIPCWKM-chp4}{
\widehat{\theta}_{\mathrm{IPCW-KM}}(d_1,d_{0}) = \int_{0}^{\tau}\widehat{S}^{(d_1)}_{\mathrm{IPCW}}(t)-\widehat{S}^{(d_0)}_{\mathrm{IPCW}}(t)dt.
}\end{equation} 

IPCW estimation relies on correct specification of the censoring weights (see, e.g., \cite{voinot2026treatmenteffectestimationcausal}). Consistency and asymptotic normality require correct specification of both the natural censoring process $C$ and the artificial censoring process $G^{(d)}$. Since the framework involves two distinct censoring mechanisms, misspecification of either component may induce bias.

These limitations motivate alternative estimators, such as outcome regression or doubly robust methods, which may provide improved robustness to model misspecification.

\paragraph{Note.} The IPCW estimator after cloning--censoring can be viewed as a particular instance of inverse probability weighting (IPW) in a time-varying setting. While IPW usually refers to weighting by treatment assignment probabilities, here the weights are based on the probability of remaining uncensored over time. In this context, artificial censoring induced by deviations from the treatment strategy can be interpreted as a time-varying treatment assignment process, explaining the close connection with IPW \citep{Burger2024,Truong2024}. 

\subsection{G-formula}\label{sec-baseline-gf}

Under Assumptions~\ref{ass:adm-treat-consistency} (\nameref{ass:adm-treat-consistency}) and 
\ref{ass:adm-cens-consistency} (\nameref{ass:adm-cens-consistency}), and Proposition~\ref{prp:seq-cens-riskset}, the strategy-specific survival function satisfies
\[
S^{(d)}(t)
=
\mathbb P(T^{(d)}>t)
=
\mathbb E\!\left[
\mathbb P(T^{(d)}>t \mid X, X_c)
\right],
\]
where the adjustment set now includes both baseline confounders $X$ and 
censoring-related covariates $X_c$. 
The corresponding covariate-specific RMST is
\[
\mu(X,X_c,d)
=
\int_0^\tau P(T^{(d)}>t \mid X, X_c)\,\mathrm dt,
\]
and, by Fubini's theorem, the causal contrast can be written as
\[
\theta_{\mathrm{RMST}}(d_1,d_0)
=
\mathbb E\!\left[
\mu(X,X_c,d_1)-\mu(X,X_c,d_0)
\right].
\]

\begin{definition}[G-formula estimator]
Let $\widehat \mu(X,X_c,d)$ be an estimator of 
$\mu(X,X_c,d)$ obtained by outcome regression in the cloned data. 
The plug-in G-formula estimator is
\[
\widehat{\theta}_{\mathrm{G}}(d_1,d_0)
=
\frac{1}{n}\sum_{i=1}^n
\left\{
\widehat \mu(X_i,X_{c,i},d_1)\,\mathrm dt
-
\widehat \mu(X_i,X_{c,i},d_0)\,\mathrm dt
\right\}.
\]
\end{definition}

A key practical advantage of the G-formula in the post--cloning--censoring framework is that it does not require explicit modeling of the artificial or natural censoring mechanisms. Under conditionally independent censoring, it reduces to a standard outcome regression requiring adjustment for all variables related to both natural and artificial censoring. Standard asymptotic properties follow from \cite{voinot2026treatmenteffectestimationcausal}, the only difference being the enlarged adjustment set $(X,X_c)$.

To reduce reliance on either the outcome model or the censoring model, we now consider doubly robust approaches combining outcome regression and inverse probability weighting. In particular, augmented inverse probability of censoring weighting (AIPCW) remains consistent if at least one nuisance model is correctly specified.

\subsection{AIPCW (Augmented IPCW)}\label{sec:aipcw-baseline}
Building on the IPCW and G-formula formulations above and under Assumptions~\ref{ass:adm-treat-consistency}, 
\ref{ass:adm-cens-consistency}, and Proposition~\ref{prp:seq-cens-riskset}, we consider an augmented inverse probability of censoring weighted (AIPCW) approach improving robustness to misspecification.

Let $W_i^{(d)}(t)$ be the inverse probability of censoring weight defined in either Equation~\eqref{eq:ipcw-art} or \eqref{eq:ipcw-tot}. For $t\in[0,\tau]$, define the conditional expected residual survival time as
\[
Q^{(d)}(t\mid X_i)
=
\mathbb E\!\left[T^{(d)}_i\wedge\tau \mid X_i,\,T_i^{(d)}>t\right]
=
\frac{1}{S^{(d)}(t\mid X_i)}
\int_t^\tau S^{(d)}(u\mid X_i)\,\mathrm du,
\]
where $S^{(d)}(t\mid X_i)=\mathbb P(T^{(d)}_i>t\mid X_i)$ denotes the strategy-specific conditional survival function. Following the doubly robust transformation proposed by \cite{rubin2007}, define
\begin{equation}
\label{eq:TDR-baseline-Wtot}
\begin{aligned}
T^{(d)}_{\mathrm{DR},i}
=
&\;
\Delta_i^{(d)}\,\widetilde T_i^{(d)}\,
W_i^{(d)}(\widetilde T_i^{(d)})
+
(1-\Delta_i^{(d)})\,
Q^{(d)}(\widetilde T_i^{(d)} \mid X_i, X_{c,i})\,
W_i^{(d)}(\widetilde T_i\wedge\tau) \\
&-
\int_0^{\widetilde T_i\wedge\tau}
\bigl(W_i^{(d)}(t)\bigr)^2
\,Q^{(d)}(t\mid X_i,X_{c,i})\,
\mathrm d\mathbb P_{H^{(d)}}(t\mid X_i,X_{c,i}),
\end{aligned}
\end{equation}
where $\mathrm d\mathbb P_{H^{(d)}}(t\mid X_i,X_{c,i})$ denotes the conditional distribution of the total censoring time. When the total censoring process admits intensity
$\lambda_{H^{(d)}}(t|X_i,X_{c,i})
= \lambda_{G^{(d)}}(t|X_i)
+ \lambda_{C}(t|X_{c,i})$,
the correction term can equivalently be written as
\[
-\int_0^{\widetilde T_i^{(d)}}
W_i^{(d)}(t)\,
Q^{(d)}(t\mid X_i,X_{c,i})\,
\lambda_{H^{(d)}}(t\mid X_i,X_{c,i})\,
\mathrm dt,
\]
making explicit the contribution of the instantaneous censoring mechanism.

This transformation IPCW-weights observed event times, imputes censored outcomes through $Q$, and adds an augmentation term correcting residual bias. Averaging $T^{(d)}_{\mathrm{DR},i}$ over clones assigned to strategy $s$ yields an augmented inverse probability of censoring weighted estimator of the RMST under strategy $s$.

\begin{definition}[AIPCW estimator]\label{def:aipcw-baseline}
Let $\hat{T}_{\mathrm{DR}}$ be an estimator of 
$T^*_{\mathrm{DR}}$ in the cloned data. 
The corresponding RMST estimator is
\[
\widehat\theta_{\mathrm{AIPCW}}(d_1,d_0) := \frac1n \sum_{i=1}^n \hat{T}^{(d_1)}_{\mathrm{DR},i}- \hat{T}^{(d_0)}_{\mathrm{DR},i}
\]
\end{definition}

Consistency is guaranteed if either the censoring model underlying $W_i^{(d)}$ or the conditional truncated expectation $Q^{(d)}$ is correctly specified.

\paragraph{Practical considerations.}
Double robustness provides protection against misspecification of the censoring mechanisms and naturally extends IPCW estimators for RMST estimation under immortal time bias. However, implementation may be computationally demanding because several nuisance components must be estimated. In practice, careful diagnostics and pragmatic modeling choices are required to balance robustness and feasibility. 

\subsection{Numerical experiments}\label{sec-simulation_baseline}

We investigate challenges arising from the coexistence of natural and artificial censoring. We assess the impact of misspecifying each censoring process, examine whether they should be modeled jointly or separately, and compare alternative approaches such as the G-formula. The analysis focuses on the causal contrast between discontinuation at three years and continuation until five years.

\paragraph{Data-generating process.}
Data are generated according to the mechanism summarized in
Table~\ref{tab:dgp} for \(n\) individuals, with parameters chosen to mimic our practical application. Individuals are followed over a discrete time grid with \(K=4\) visits indexed by \(k=0,\dots,4\) 
, and follow-up is administratively truncated at \(\tau=10\), corresponding to the 10-year horizon of the Breast Cancer study.
At each visit, event and censoring times are generated on a continuous time scale within visit intervals.

\begin{table}[ht]
\centering
\caption{Data-generating processes used in the simulation study.}
\label{tab:dgp}
\renewcommand{\arraystretch}{1.15}
\setlength{\tabcolsep}{5pt}
\small
\begin{tabular}{p{2.5cm} p{12.7cm}}
\toprule
\textbf{Scenario} & \textbf{Data-generating mechanism / parameters} \\
\midrule

\textbf{General DGP}
&
\(
X_1 \sim \mathcal N(\mu_1,\sigma_1^2),\;
X_2 \sim \mathcal N(\mu_2,\sigma_2^2),\;
X_3 \sim \mathcal N(\mu_3,\sigma_3^2)
\)
\\[3pt]
&
for \(k \in \{0,1,\dots,K\}\),
\\[-1pt]
&
\(
A_0=1,\quad
\mathbb P(A_k=1\mid A_{k-1},X)
=
\operatorname{logit}^{-1}\!\big(
\gamma_0+\gamma_A A_{k-1}+\gamma_{X1}X_1+\gamma_{X2}X_2+\gamma_t(k)
\big),
\quad
\gamma_t(k)=
\left\{
\begin{array}{ll}
\gamma_{t,\mathrm{pre}}, & k \le 2,\\
\gamma_{t,\mathrm{post}}, & k > 2
\end{array}
\right.
\)
\\[5pt]
&
\(
\lambda_T(k\mid A_k,X)=
\exp\!\big(
\alpha_0+\alpha_A A_k+\alpha_{X1}X_1+\alpha_{X2}X_2+\alpha_{X3}X_3+\alpha_t(k)
\big),
\quad
\alpha_t(k)=
\left\{
\begin{array}{ll}
\alpha_{t,\mathrm{pre}}, & k \le 2,\\
\alpha_{t,\mathrm{post}}, & k > 2
\end{array}
\right.
\)
\\[5pt]
&
\(
\lambda_C(k\mid A_k,X)=
\beta\exp\!\big(
\beta_0+\beta_A A_k+\beta_{X2}X_2+\beta_{X3}X_3+\beta_t(k)
\big),
\quad
\beta_t(k)=
\left\{
\begin{array}{ll}
\beta_{t,\mathrm{pre}}, & k \le 2,\\
\beta_{t,\mathrm{post}}, & k > 2
\end{array}
\right.
\)
\\
\bottomrule
\end{tabular}
\end{table}
To promote positivity, we introduced additional parameters through $\gamma(k)$, $\alpha(k)$, and $\beta(k)$ unrelated to baseline covariates. In particular, $\gamma(k)$ was reduced after the second visit so that enough individuals followed the discontinuation strategy, while $\alpha(k)$ was increased after the second visit to raise the event hazard, but kept lower beforehand so that sufficient individuals remained at risk. These time-varying choices ensured adequate support for the target strategies without introducing a grace period, keeping the simulation setting simpler since immortal time bias is already induced by the longitudinal treatment strategy itself.

From this general data-generating process (DGP), we derive four scenarios ranging from low to high confounding and from low to highly informative natural censoring (parameter values are reported in Appendix~\ref{appendix_chp4} in Section~\ref{sec-coef_baseline}).

\paragraph{Methods.}
We first evaluated two naive approaches: 
(i) a complete-case analysis restricted to individuals naturally following the target duration strategies and ignoring censoring, and 
(ii) a naive G-formula applied only to individuals naturally following the target strategies.

We then applied the cloning procedure described in Figure~\ref{fig:cloning}. On the cloned data, RMST differences were estimated using Kaplan--Meier, G-formula with Cox models, and IPCW. Artificial censoring was modeled in discrete time using pooled logistic regression, whereas natural censoring was modeled using a piecewise exponential model implemented through Poisson regression with a log-offset for interval length. In both cases, one model was fitted separately within each treatment arm on person-period data (Figure~\ref{fig:simu_illustration}), so that uncensoring probabilities were updated on visit-defined intervals and used at each event time through interval-specific weights.

\begin{figure}[h]
    \centering
    \includegraphics[width=0.9\linewidth]{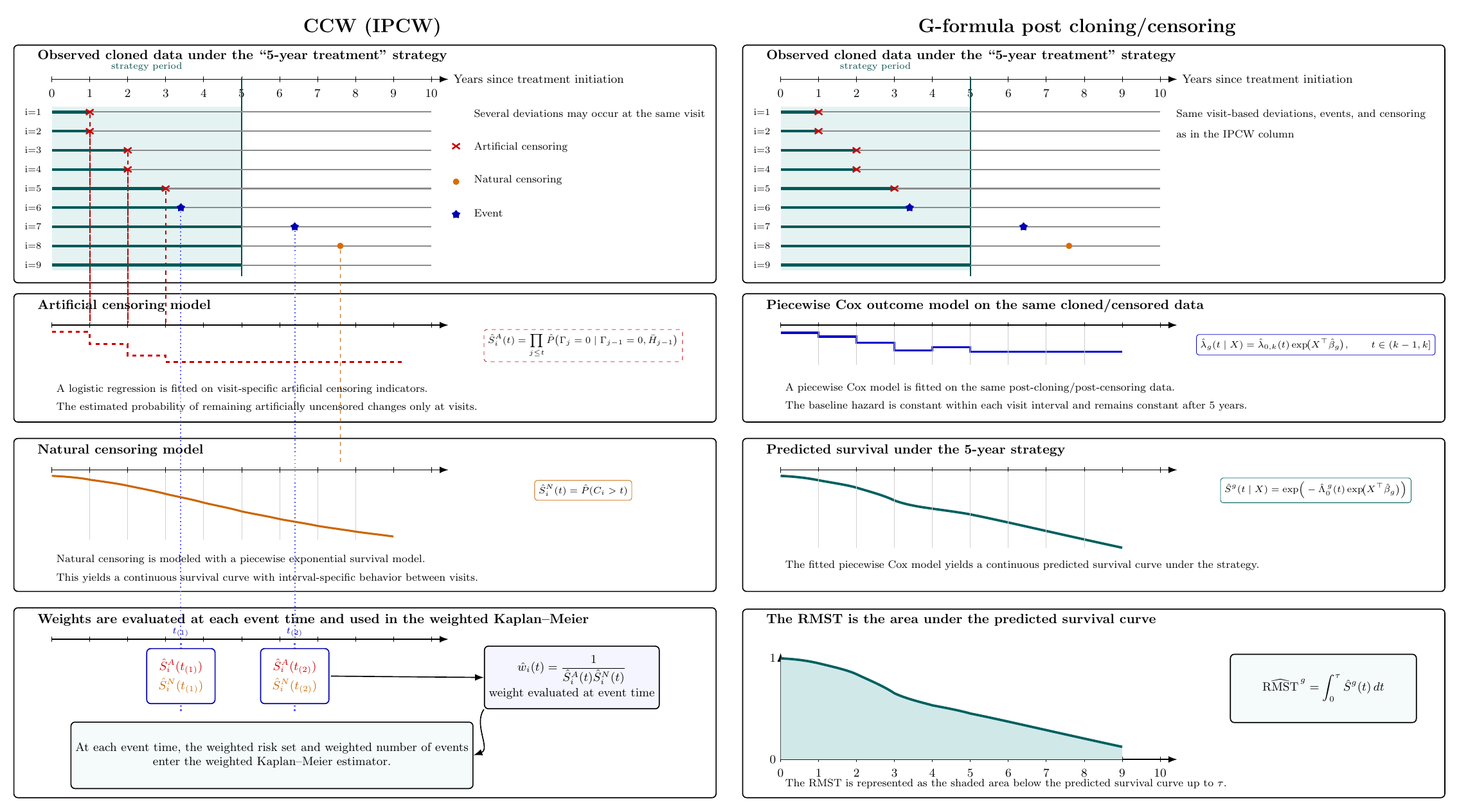}
    \caption{Illustration of IPCW and G-formula computation after cloning and censoring in the simulation study.}
    \label{fig:simu_illustration}
\end{figure}


To evaluate the consequences of model misspecification, we also considered a single regression model using either Pooled logistic or a Piecewise Exponential model to capture the overall censoring process. 



The results are summarized using boxplots across different sample sizes (\(n=500,1000,2000,4000,8000\)). 
Each boxplot represents the distribution of Monte Carlo estimates obtained from \(50\) independent replicates for each scenario and sample size. 
In Appendix, Table~\ref{tab:rmse_baseline} shows the bias, ESE, MSE and RMSE of the different estimator in one scenario (bottom right boxplot in Figure~\ref{fig:ate_baseline}, scenario 4: strong confounding and dependent censoring). 

\begin{figure}[H]
    \centering
    \includegraphics[width=1\linewidth]{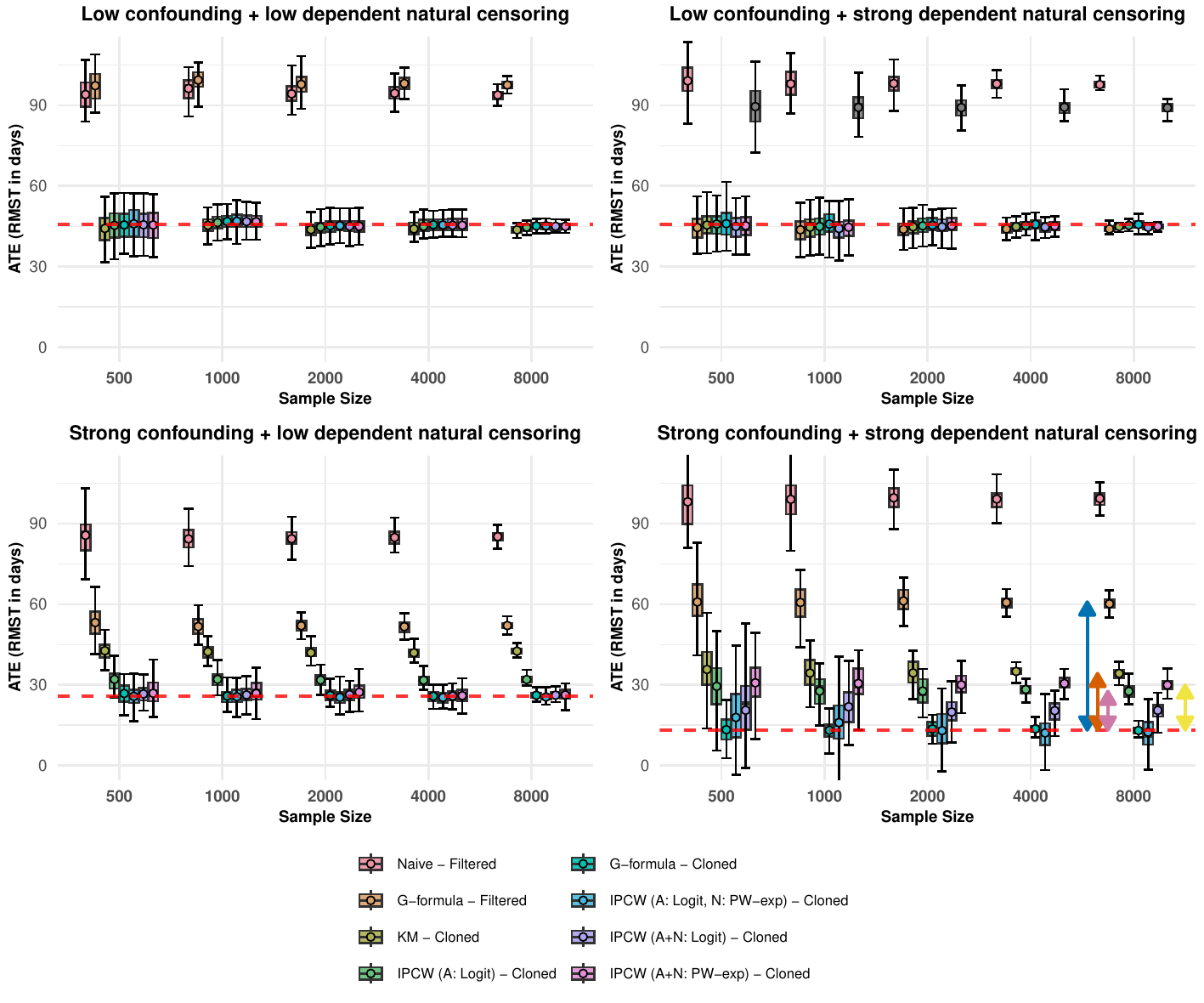}
    \caption{RMST difference estimates after cloning and artificial censoring under four simulation settings with baseline covariates, defined by combinations of low or high confounding and low or high dependent natural censoring. In estimator labels, A and N indicate the models used for artificial and natural censoring, respectively.}
    \label{fig:ate_baseline}
\end{figure}

Across the four scenarios, several patterns emerge.

First, filtering individuals based on observed treatment duration without accounting for immortal time bias leads to severe bias (Naive - Filtered), even when baseline confounding and informative natural censoring are partially addressed through a naive outcome regression restricted to adherent individuals (G-formula - Filtered). This corresponds to the \textcolor{navy_cbf1}{blue arrow}
(Figure~\ref{fig:ate_baseline} bottom right panel, 8000 individuals). 

Second, after cloning removes immortal time bias, failing to account for artificial and natural censoring still induces bias, illustrated by the \textcolor{red_cbf1}{red arrow} (KM -- Cloned).

Third, adjusting for natural censoring becomes increasingly important as censoring informativeness strengthens. When natural censoring depends more strongly on prognosis, ignoring it (IPCW (A: Logit) -- Cloned) leads to increasing bias even after cloning. This additional bias is represented by the \textcolor{purple_cbf1}{purple arrow} in Figure~\ref{fig:ate_baseline}.

When both artificial and natural censoring are correctly modeled (IPCW (A: Logit, N: PW--exp) - Cloned), the IPCW estimator performs well, although substantial variability appears in the most challenging scenario because of extreme weights.

In contrast, the G-formula after cloning (G-formula - Cloned) performs well across all scenarios and remains stable even under strong confounding or highly informative censoring. This stability may partly stem from the exponential outcome model used in the data-generating process, whose parametric structure facilitates extrapolation in sparsely supported covariate regions and may therefore favor outcome regression approaches. To investigate this point, we conducted additional simulations using a Weibull outcome model for the G-formula. In this setting, substantial bias remains even with 8,000 observations (Table~\ref{tab:rmse_baseline}), illustrating the sensitivity of the estimator to model misspecification.

Finally, modeling the two censoring processes through a single regression model (IPCW (A+N: Logit) -- Cloned and IPCW (A+N: PW--exp) -- Cloned) induces scenario-dependent bias. In more complex settings, estimation relies more heavily on the weights, making misspecification more impactful. The \textcolor{yellow_cbf2}{yellow arrow} illustrates this bias. While the single logistic regression often captures both censoring processes reasonably well, the single Cox model appears systematically biased in our settings. Using one model for both censoring processes is therefore risky overall, although it may reduce variance by avoiding products of separately estimated censoring probabilities that can generate extreme weights.

\section{Causal effect of treatment duration with time-dependent confounders}\label{sec-tdconf}

Scenario~(2), which allows for time-dependent confounders, represents a more realistic setting and is summarized by the DAG below. The corresponding identification assumptions are recalled in Table~\ref{tab:assumptions_timedep}. Because treatment decisions, censoring, and covariates now evolve sequentially over follow-up, identification relies on sequential conditional independence assumptions. This, in turn, requires formulating the problem in a discrete-time framework.

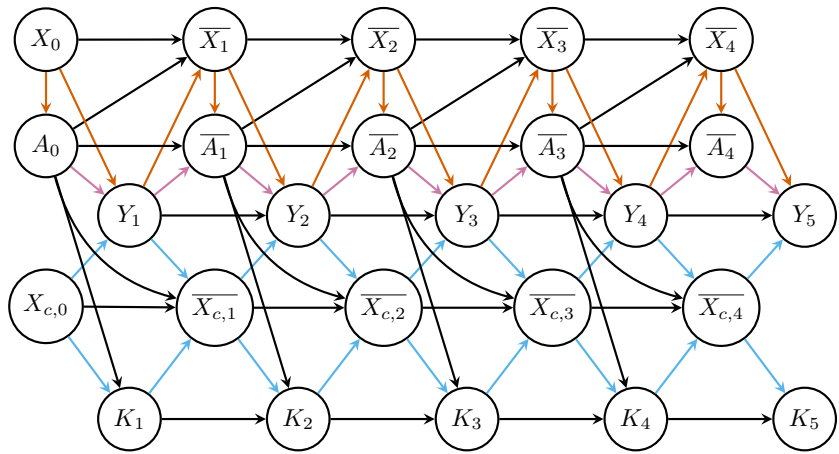
\begin{figure}[H]
    \centering
\resizebox{0.65\textwidth}{!}{
\begin{tikzpicture}[
    ->,
    >=stealth,
    node distance=1.1cm and 1.0cm,
    thick,
    scale=0.9,
    every node/.style={transform shape}
]
    \def\r{0.9cm}
    
    \node (A2) [draw, circle, minimum size=\r, below right=2.5cm and 0.3cm of X] {$\overline{A_2}$};
    \node (A1) [draw, circle, minimum size=\r, left=1.5cm of A2] {$\overline{A_1}$};
    \node (A0) [draw, circle, minimum size=\r, left=1.5cm of A1] {$A_0$};
    \node (A3) [draw, circle, minimum size=\r,  right=1.5cm of A2] {$\overline{A_3}$};
    \node (A4) [draw, circle, minimum size=\r,  right=1.5cm of A3] {$\overline{A_4}$};

    \node (Y3) [draw, circle, minimum size=\r,  below right=3.5cm and 1.5cm of X] {$Y_3$};
    \node (Y2) [draw, circle, minimum size=\r, left=1.5cm of Y3] {$Y_2$};
    \node (Y1) [draw, circle, minimum size=\r, left=1.5cm of Y2] {$Y_1$};
    \node (Y4) [draw, circle, minimum size=\r,  right=1.5cm of Y3] {$Y_4$};
    \node (Y5) [draw, circle, minimum size=\r,  right=1.5cm of Y4] {$Y_5$};

    \node (K3) [draw, circle, minimum size=\r,  below=2cm of Y3] {$K_3$};
    \node (K2) [draw, circle, minimum size=\r, left=1.5cm of K3] {$K_2$};
    \node (K1) [draw, circle, minimum size=\r, left=1.5cm of K2] {$K_1$};
    \node (K4) [draw, circle, minimum size=\r,  right=1.5cm of K3] {$K_4$};
    \node (K5) [draw, circle, minimum size=\r,  right=1.5cm of K4] {$K_5$};

    \node (L3) [draw, circle, minimum size=\r,  above=0.6cm of A3] {$\overline{X_3}$};
    \node (L2) [draw, circle, minimum size=\r, above=0.6cm of A2] {$\overline{X_2}$};
    \node (L0) [draw, circle, minimum size=\r, above=0.6cm of A0] {$X_0$};
    \node (L1) [draw, circle, minimum size=\r, above=0.6cm of A1] {$\overline{X_1}$};
    \node (L4) [draw, circle, minimum size=\r, above=0.6cm of A4] {$\overline{X_4}$};

    \node (Xc3) [draw, circle, minimum size=\r,  below=1.3cm of A3] {$\overline{X_{c,3}}$};
    \node (Xc2) [draw, circle, minimum size=\r, below=1.3cm of A2] {$\overline{X_{c,2}}$};
    \node (Xc0) [draw, circle, minimum size=\r, below=1.3cm of A0] {$X_{c,0}$};
    \node (Xc1) [draw, circle, minimum size=\r, below=1.3cm of A1] {$\overline{X_{c,1}}$};
    \node (Xc4) [draw, circle, minimum size=\r, below=1.3cm of A4] {$\overline{X_{c,4}}$};

    \draw[red_cbf2]  (L0) -- (A0);
    \draw[red_cbf2]  (L1) -- (A1);
    \draw[red_cbf2]  (L2) -- (A2);
    \draw[red_cbf2]  (L3) -- (A3);
    \draw[red_cbf2]  (L4) -- (A4);

    \draw (A0) -- (L1);
    \draw (A1) -- (L2);
    \draw (A2) -- (L3);
    \draw (A3) -- (L4);
    
    \draw (L0) -- (L1);
    \draw (L1) -- (L2);
    \draw (L2) -- (L3);
    \draw (L3) -- (L4);

    \draw[red_cbf2] (L0) -- (Y1);
    \draw[red_cbf2] (L1) -- (Y2);
    \draw[red_cbf2] (L2) -- (Y3);
    \draw[red_cbf2] (L3) -- (Y4);
    \draw[red_cbf2] (L4) -- (Y5);

    \draw[red_cbf2](Y1) -- (L1);
    \draw[red_cbf2] (Y2) -- (L2);
    \draw[red_cbf2] (Y3) -- (L3);
    \draw[red_cbf2] (Y4) -- (L4);

    \draw[purple_cbf2] (A0) -- (Y1);
    \draw[purple_cbf2] (Y1) -- (A1);
    \draw[purple_cbf2] (A1) -- (Y2);
    \draw[purple_cbf2] (Y2) -- (A2);
    \draw[purple_cbf2] (A2) -- (Y3);
    \draw[purple_cbf2] (Y3) -- (A3);
    \draw[purple_cbf2] (A3) -- (Y4);
    \draw[purple_cbf2] (Y4) -- (A4);
    \draw[purple_cbf2] (A4) -- (Y5);

    \draw (Xc0) -- (Xc1);
    \draw (Xc1) -- (Xc2);
    \draw (Xc2) -- (Xc3);
    \draw (Xc3) -- (Xc4);

    \draw[blue_cbf2] (Y1) to (Xc1);
    \draw[blue_cbf2] (Y2) to (Xc2);
    \draw[blue_cbf2] (Y3) to (Xc3);
    \draw[blue_cbf2] (Y4) to (Xc4);
    
    \draw[blue_cbf2] (Xc0) -- (Y1);
    \draw[blue_cbf2] (Xc1) -- (Y2);
    \draw[blue_cbf2] (Xc2) -- (Y3);
    \draw[blue_cbf2] (Xc3) -- (Y4);
    \draw[blue_cbf2] (Xc4) -- (Y5);

    \draw[blue_cbf2] (Xc0) -- (K1);
    \draw[blue_cbf2] (Xc1) -- (K2);
    \draw[blue_cbf2] (Xc2) -- (K3);
    \draw[blue_cbf2] (Xc3) -- (K4);
    \draw[blue_cbf2] (Xc4) -- (K5);

    \draw [bend right = 30] (A0) to (Xc1);
    \draw [bend right = 30] (A1) to (Xc2);
    \draw [bend right = 30] (A2) to (Xc3);
    \draw [bend right = 30] (A3) to (Xc4);

    \draw [bend right = 0] (A0) to (K1);
    \draw [bend right = 0] (A1) to (K2);
    \draw [bend right = 0] (A2) to (K3);
    \draw [bend right = 0] (A3) to (K4);

    \draw[blue_cbf2] (K1) -- (Xc1);
    \draw[blue_cbf2] (K2) -- (Xc2);
    \draw[blue_cbf2] (K3) -- (Xc3);
    \draw[blue_cbf2] (K4) -- (Xc4);

    \draw (Y1) -- (Y2);
    \draw (Y2) -- (Y3);
    \draw (Y3) -- (Y4);
    \draw (Y4) -- (Y5);

    \draw (A0) -- (A1);
    \draw (A1) -- (A2);
    \draw (A2) -- (A3);
    \draw (A3) -- (A4);

    \draw (K1) -- (K2);
    \draw (K2) -- (K3);
    \draw (K3) -- (K4);
    \draw (K4) -- (K5);

\end{tikzpicture}
}
\caption{Discrete-time representation of the treatment assignment process $A_k$, the event process
$Y$, and the natural censoring process $K$, with $X_t$ denoting time-varying confounders and $X_{c,t}$ time-varying censoring-related covariates. \textcolor{blue_cbf2}{Blue edges} represent informative natural censoring, while \textcolor{red_cbf2}{red edges} confounding effect and \textcolor{purple_cbf2}{purple edges} the immortal time bias.}
    \label{fig:dag_discrete2}
\end{figure}
Figure~\ref{fig:dag_discrete2} depicts a discrete-time causal DAG with time-varying processes, in which treatment decisions $A_k$, time-dependent covariates $X_k$, and censoring-related variables $X_{c,k}$ are updated over follow-up. 
Compared with Scenario~(1), this setting is more challenging because $X_k$ are time-varying confounders of the treatment--outcome relationship while also being affected by prior treatment. Consequently, treatment--covariate feedback complicates the evaluation of duration strategies and requires methods that appropriately adjust for time-varying confounding.

\subsection{Cloning-Censoring}

We refer to Section~\ref{sec-cloningcensoring} for methodological details, as the cloning and artificial censoring procedure is identical across scenarios. 
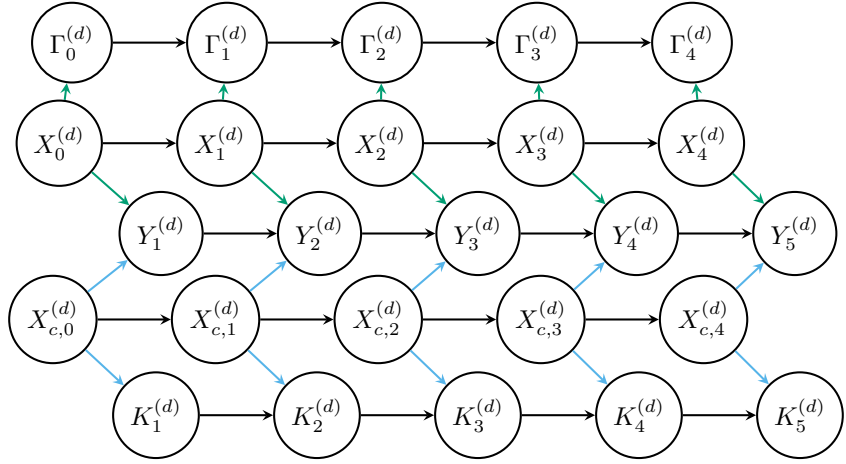
\begin{figure}[H]
    \centering
    \resizebox{0.65\textwidth}{!}{
\begin{tikzpicture}[->, >=stealth, node distance=1.1cm and 1.0cm, thick]
    \def\r{0.85cm}
    
    \node (Y3) [draw, circle, minimum size=\r] {$Y_3^{(d)}$};
    \node (Y2) [draw, circle, minimum size=\r, left=1.0cm of Y3] {$Y_2^{(d)}$};
    \node (Y1) [draw, circle, minimum size=\r, left=1.0cm of Y2] {$Y_1^{(d)}$};
    \node (Y4) [draw, circle, minimum size=\r,  right=1.0cm of Y3] {$Y_4^{(d)}$};
    \node (Y5) [draw, circle, minimum size=\r,  right=1.0cm of Y4] {$Y_5^{(d)}$};

    \node (G3) [draw, circle, minimum size=\r,  above right=1.8cm and 0cm of Y3] {$\Gamma_3^{(d)}$};
    \node (G2) [draw, circle, minimum size=\r, left=1.0cm of G3] {$\Gamma_2^{(d)}$};
    \node (G1) [draw, circle, minimum size=\r, left=1.0cm of G2] {$\Gamma_1^{(d)}$};
    \node (G0) [draw, circle, minimum size=\r,  left=1.0cm of G1] {$\Gamma_0^{(d)}$};
    \node (G4) [draw, circle, minimum size=\r,  right=1.0cm of G3] {$\Gamma_4^{(d)}$};

    
    \node (K3) [draw, circle, minimum size=\r,  below=1.3cm of Y3] {$K_3^{(d)}$};
    \node (K2) [draw, circle, minimum size=\r, left=1.0cm of K3] {$K_2^{(d)}$};
    \node (K1) [draw, circle, minimum size=\r, left=1.0cm of K2] {$K_1^{(d)}$};
    \node (K4) [draw, circle, minimum size=\r,  right=1.0cm of K3] {$K_4^{(d)}$};
    \node (K5) [draw, circle, minimum size=\r,  right=1.0cm of K4] {$K_5^{(d)}$};

    \node (L2) [draw, circle, minimum size=\r,  above left=0.4cm and 0.5cm of Y3] {$X_2^{(d)}$};
    \node (L1) [draw, circle, minimum size=\r,  left=1.0cm of L2] {$X_1^{(d)}$};
    \node (L0) [draw, circle, minimum size=\r,  left=1.0cm of L1] {$X_0^{(d)}$};
    \node (L3) [draw, circle, minimum size=\r, right=1.0cm of L2] {$X_3^{(d)}$};
    \node (L4) [draw, circle, minimum size=\r, right=1.0cm of L3] {$X_4^{(d)}$};

    \node (Xc3) [draw, circle, minimum size=\r,  below =1.2cm of L3] {$X^{(d)}_{c,3}$};
    \node (Xc2) [draw, circle, minimum size=\r, left=1.0cm of Xc3] {$X^{(d)}_{c,2}$};
    \node (Xc1) [draw, circle, minimum size=\r, left=1.0cm of Xc2] {$X^{(d)}_{c,1}$};
    \node (Xc0) [draw, circle, minimum size=\r,  left=1.0cm of Xc1] {$X^{(d)}_{c,0}$};
    \node (Xc4) [draw, circle, minimum size=\r,  right=1.0cm of Xc3] {$X^{(d)}_{c,4}$};

    \draw (L0) -- (L1);
    \draw (L1) -- (L2);
    \draw (L2) -- (L3);
    \draw (L3) -- (L4);

    \draw (Y1) -- (Y2);
    \draw (Y2) -- (Y3);
    \draw (Y3) -- (Y4);
    \draw (Y4) -- (Y5);

    \draw (K1) -- (K2);
    \draw (K2) -- (K3);
    \draw (K3) -- (K4);
    \draw (K4) -- (K5);

    \draw (G0) -- (G1);
    \draw (G1) -- (G2);
    \draw (G2) -- (G3);
    \draw (G3) -- (G4);

    \draw (Xc0) -- (Xc1);
    \draw (Xc1) -- (Xc2);
    \draw (Xc2) -- (Xc3);
    \draw (Xc3) -- (Xc4);


    \draw[green_cbf2] (L0) -- (Y1);
    \draw[green_cbf2] (L1) -- (Y2);
    \draw[green_cbf2] (L2) -- (Y3);
    \draw[green_cbf2] (L3) -- (Y4);
    \draw[green_cbf2] (L4) -- (Y5);

    \draw[green_cbf2] (L0) -- (G0);
    \draw[green_cbf2] (L1) -- (G1);
    \draw[green_cbf2] (L2) -- (G2);
    \draw[green_cbf2] (L3) -- (G3);
    \draw[green_cbf2] (L4) -- (G4);



    \draw[blue_cbf2] (Xc0) -- (Y1);
    \draw[blue_cbf2] (Xc1) -- (Y2);
    \draw[blue_cbf2] (Xc2) -- (Y3);
    \draw[blue_cbf2] (Xc3) -- (Y4);
    \draw[blue_cbf2] (Xc4) -- (Y5);

    \draw[blue_cbf2] (Xc0) -- (K1);
    \draw[blue_cbf2] (Xc1) -- (K2);
    \draw[blue_cbf2] (Xc2) -- (K3);
    \draw[blue_cbf2] (Xc3) -- (K4);
    \draw[blue_cbf2] (Xc4) -- (K5);

    \end{tikzpicture}}
\caption{
Discrete-time representation of treatment strategy assignment, event and censoring processes after cloning--censoring.
$\Gamma$ denotes artificial censoring process, $K$ natural censoring process, and $Y$ the
event process with $X_t$ denoting the time-varying confounders and $X_{c,t}$ the time-varying censoring-related covariates.
\textcolor{blue_cbf1}{Blue edges} represent informative natural censoring, while \textcolor{green_cbf1}{green edges} represents informative artificial censoring.}
    \label{fig:dag_cloning2}
\end{figure} 
Figure~\ref{fig:dag_cloning2} depicts the post-cloning structure associated with Figure~\ref{fig:dag_discrete2}. After cloning and artificial censoring, the resulting data structure can be viewed as a trial emulation with time-dependent informative censoring.



In the same manner than in Section~\ref{sec-cloningcensoring}, additional assumptions in this new framework have to be made to enable the identification of the causal quantity. 

\begin{proposition}[Sequential censoring exchangeability and positivity on the cloned risk set with time-varying covariates]
\label{prp:seq-cens-riskset-timedep}
Consider the cloned and artificially censored dataset. For each individual $(i,d)$, and decision time $k=1,\dots,K$, define the (observed)
at-risk indicator
\[
R^{(d)}_{k-1}=\mathbbm{1}\{\widetilde T^{(d)}>k-1\}
=\mathbbm{1}\{\min(T,C,G^{(d)})>k-1\}.
\]

Let
\[
\overline X^{(d)}_{k-1}
=
\bigl(X_{0}^{(d)}, X_{1}^{(d)}, \dots, X_{k-1}^{(d)}\bigr)
\]
denote the history of time-dependent covariates that would be observed under strategy $d$ up to time $k-1$, restricted to trajectories compatible with the strategy.

For each decision time $k$, assume that the following conditions hold among individuals in the cloned, strategy-specific risk set
$\{R^{(d)}_{i,k-1}=1\}$.

\begin{enumerate}
\item[(i)] \emph{Conditionally independent artificial censoring (risk-set version).}
Under Assumption~\ref{ass:texch-timedep}
(\nameref{ass:texch-timedep}), \ref{ass:sutva-treat-timedep}, \ref{ass:sutva-cens-timedep} and \ref{ass:sutva-cov-timedep} (SUTVA for time-dependent confounders), we assume that
\[
T^{(d)} \;\perp\!\!\!\perp\; \Gamma_{k}^{(d)}
\ \Big|\ \overline X^{(d)}_{k-1},\ R^{(d)}_{k-1}=1 .
\]

\item[(ii)] \emph{Natural censoring exchangeability.}
Under Assumptions~\ref{ass:cens-cond-timedep}
(\nameref{ass:cens-cond-timedep}) and \ref{ass:sutva-treat-timedep}, \ref{ass:sutva-cens-timedep} and \ref{ass:sutva-cov-timedep} (SUTVA for time-dependent confounders), we assume that
\[
T^{(d)} \;\perp\!\!\!\perp\; K_{k}^{(d)}
\ \Big|\ \overline X^{(d)}_{c,k-1},\ R^{(d)}_{k-1}=1 .
\]

\item[(iii)] \emph{Positivity on the risk set.}
There exists $\epsilon>0$ such that, for each $k$,
under Assumption~\ref{ass:tpos-timedep}
(\nameref{ass:tpos-timedep}),
\[
\mathbb P\!\left(\Gamma_{k}^{(d)}=0 \mid \overline X^{(d)}_{k-1},\ R^{(d)}_{k-1}=1\right)\ge \epsilon,
\]

and under Assumptions~\ref{ass:cens-cond-pos-timedep}
(\nameref{ass:cens-cond-pos-timedep}),
\[
\mathbb P\!\left(K_{k}^{(d)}=0 \mid \overline X^{(d)}_{c,k-1},\ R^{(d)}_{k-1}=1\right)\ge \epsilon.
\]
\end{enumerate}
\end{proposition}

\subsection{Inverse Probability of Censoring Weighting (IPCW): Extension to time-varying covariates}\label{sec:ipcw_timevarying}
Selection bias may arise from both artificial censoring (through deviations from the assigned duration strategy) and natural censoring, and both mechanisms may be informative through their dependence on the evolving history. The weighting principle remains unchanged: the cloned dataset is reweighted to recover the distribution that would have been observed without informative censoring.

For each strategy $g_d$, define by convention the baseline deviation indicator $\Gamma^{(d)}_{0}=0$. Let the compliance event up to decision time $k$ be
\[
\mathcal E^{(d)}_{0}=\Omega,
\qquad
\mathcal E^{(d)}_{k}
=
\Bigl\{\Gamma^{(d)}_{0}=\cdots=\Gamma^{(d)}_{k}=0\Bigr\}
=
\Bigl\{\bar A_{k}\in\mathcal A_{d,k}\Bigr\},
\quad k=1,\dots,K,
\]
and recall that $R^{(d)}_{k-1}=\mathbbm{1}\{\widetilde T^{(d)}>k-1\}$ denotes the strategy-specific at-risk indicator in the cloned dataset.

Under Proposition~\ref{prp:seq-cens-riskset-timedep}
(\nameref{prp:seq-cens-riskset-timedep}), identification is obtained by weighting each clone by the inverse probability of remaining observable under both censoring mechanisms. Define the event of remaining naturally uncensored up to $k$ as
\[
\mathcal N_{0}^{(d)}=\Omega,
\qquad
\mathcal N_{k}^{(d)}=\{K_{1}^{(d)}=0,\dots,K_{k}^{(d)}=0\},
\quad k=1,\dots,K.
\]

On the cloned risk set $\{R^{(d)}_{k-1}=1\}$, define the total censoring weight as

\begin{equation}
\label{eq:ipcw-tot-timedep}
W^{(d)}(k)
=
\frac{1}{
\mathbb P\!\left(\mathcal N_{k}^{(d)},\ \mathcal E^{(d)}_{k}
\ \big|\ \overline X^{(d)}_{k-1},\ \overline X^{(d)}_{c,k-1},\ R^{(d)}_{k-1}=1\right)
}, \qquad \ k=1,\dots,K.
\end{equation}

By the chain rule,
\begin{equation}
\label{eq:gen_weight_discrete_timedep}
\begin{aligned}
\mathbb P\!\Big(\mathcal N_{k}^{(d)},\ \mathcal E^{(d)}_{k}
\ \big|\ \overline X^{(d)}_{k-1},\ \overline X^{(d)}_{c,k-1},\ R^{(d)}_{k-1}=1\Big)
&=
\mathbb P\!\Big(\mathcal E^{(d)}_{k}
\ \big|\ \overline X^{(d)}_{k-1},\ \overline X^{(d)}_{c,k-1},\ R^{(d)}_{k-1}=1\Big) \\
&\quad \times
\mathbb P\!\Big(\mathcal N_{k}^{(d)}
\ \big|\ \overline X^{(d)}_{k-1},\ \overline X^{(d)}_{c,k-1},\ \mathcal E^{(d)}_{k},\ R^{(d)}_{k-1}=1\Big).
\end{aligned}
\end{equation}
which does not require independence between artificial and natural censoring.

Under additional conditional independence assumptions, these conditional models simplify to
\[
\mathbb P\!\left(\mathcal E^{(d)}_{k}
\ \big|\ \overline X^{(d)}_{k-1},\ \overline X^{(d)}_{c,k-1},\ R^{(d)}_{k-1}=1\right)
=
\mathbb P\!\left(\mathcal E^{(d)}_{k}
\ \big|\ \overline X^{(d)}_{k-1},\ R^{(d)}_{k-1}=1\right),
\]
and
\[
\mathbb P\!\left(\mathcal N_{k}^{(d)}
\ \big|\ \overline X^{(d)}_{k-1},\ \overline X^{(d)}_{c,k-1},\ \mathcal E^{(d)}_{k},\ R^{(d)}_{k-1}=1\right)
=
\mathbb P\!\left(\mathcal N_{k}^{(d)}
\ \big|\ \overline X^{(d)}_{c,k-1},\ R^{(d)}_{k-1}=1\right).
\]
These are modeling assumptions rather than identification requirements.

\begin{definition}[IPCW--Kaplan--Meier with time-varying covariates]
\protect\hypertarget{def-ipcwkm-chp4_timedep}{}\label{def-ipcwkm-chp4_timedep}
Let $W_{i}^{(d)}(\cdot)$ be defined as in
Equation~\eqref{eq:ipcw-tot-timedep}. 
The IPCW--Kaplan--Meier estimator of the counterfactual survival function under strategy $g_d$ is defined as in Definition~\ref{def-ipcwkm-chp4}, replacing the weights by $W_{i}^{(d)}$ evaluated at observed event times.
\end{definition}

We show in Proposition~\ref{prp-ipcwkm_multiply} (Appendix~\ref{appendix_chp4}) that the oracle version of this estimator with weights defined as in Equation~\eqref{eq:ipcw-tot-timedep} is strongly consistent and asymptotically normal.

\paragraph{Practical considerations with time-varying confounders.}
When treatment assignment and censoring depend on time-varying covariates, artificial censoring becomes intrinsically linked to the evolving treatment and covariate history. Identification then relies on sequential conditional exchangeability assumptions and estimation requires a counting-process or person--period representation of the data.

Because treatment and censoring models depend on expanding histories, nuisance-model dimensionality increases over time, amplifying the risks of misspecification and practical positivity violations in poorly supported regions of the covariate-history space. These features may exacerbate weight instability and finite-sample variability of CCW--IPCW estimators. An illustrative implementation of the IPCW estimator under independent natural censoring is provided in Algorithm~\ref{alg:ipcw_artificial_natural} (Appendix~\ref{appendix_chp4}) using the Breast Cancer dataset.
\subsection{G-formula}

After cloning and artificial censoring, longitudinal covariates remain post-baseline variables and typically lie on the causal pathway between the treatment strategy and survival. As a result, the covariate trajectory becomes \emph{strategy-dependent}. The G-formula can be written as follows.

\begin{definition}[G-formula for RMST under a strategy]

The RMST contrast is denoted as
\[
\theta_{\mathrm{RMST}}(d_1,d_0)
=
\psi(d_1)-\psi(d_0),
\]

where, for a given treatment strategy \(d\) and for $\mathbb{X}_k=(X_k,X_{c,k})$,

\begin{align*}
\psi(d)
=
\mathbb{E}\!\left[T^{(d)}\right]
&=
\sum_{k \geq 0}
\mathbb{P}\!\left(Y_k^{(d)}=0\right) \\
&=
\sum_{k \geq 0}
\int
\mathbb{P}\!\left(
Y_{k+1}=0
\mid
\bar A_k \in \bar{\mathcal A}_k^{(d)},\,
\bar{\mathbb X}_k=\bar x_k
\right)
\prod_{l \leq k}
f_d(\bar x_l \mid \bar x_{l-1})
\,d\bar x_l .
\end{align*}

with
\[
f_d(\bar x_l \mid \bar x_{l-1})
=
\mathbb{P}\!\left(
\bar{\mathbb X}_l=\bar x_l
\mid
\bar{\mathbb X}_{l-1}=\bar x_{l-1},\,
\bar A_{l-1}\in\bar{\mathcal A}_{l-1}^{(d)},\,
Y_l=0
\right).
\] 
denoting the joint distribution of the covariate trajectories under strategy \(d\).
\end{definition}

This representation highlights that the estimand involves the full distribution of \emph{strategy-specific} covariate trajectories. Consequently, a plug-in implementation of the G-formula cannot rely only on observed covariate histories: each individual contributes at most one realized trajectory, whereas the estimand requires counterfactual trajectories $\bar{\mathbb{X}}^{(d)}$. 

\paragraph{Practical considerations.}
Estimation therefore requires modeling the covariate process under intervention and propagating it over follow-up. In principle, this can be achieved through Robins' parametric G-computation algorithm \citep{robins1986new} or alternative semi- or nonparametric approaches. In our implementation (Algorithm~\ref{alg:gformula_parametric_general} in Appendix~\ref{appendix_chp4}), however, we do not model the full joint longitudinal process as in the classical parametric G-formula; instead, we reconstruct selected time-varying covariate trajectories under each target strategy and combine them with an outcome model fitted on the observed data. More generally, these approaches rely on nuisance models for both the outcome and covariate evolution and may therefore be sensitive to model misspecification.

In practice, covariate trajectories under a given strategy are typically approximated through Monte Carlo simulation by generating multiple plausible trajectories for each individual. Estimation errors may accumulate over time: similarly to IPCW, where errors in censoring probabilities propagate multiplicatively through the weights, inaccuracies in the modeled covariate process can propagate along simulated trajectories and induce substantial bias.

\paragraph{Note.} To our knowledge, no doubly robust AIPCW estimator has been specifically developed for settings with time-varying covariates within a cloning--censoring--weighting framework. More broadly, doubly robust approaches such as targeted minimum loss-based estimation (TMLE) rely on flexible estimation of nuisance functions for the outcome, treatment, and censoring mechanisms, often through data-adaptive procedures such as Super Learner. Dedicated software such as the \href{https://cran.r-project.org/web/packages/ltmle/index.html}{\texttt{ltmle}} R package \citep{ltmle} is available for longitudinal settings with treatments, time-varying covariates, and right censoring. However, similarly to Robins' G-formula \citep{robins1986new}, existing TMLE approaches are generally formulated for discrete-time longitudinal settings and are not designed off the shelf for cloning--censoring--weighting frameworks with grace periods and artificial censoring. Although methods such as \cite{Stitelman2012} provide theoretical guarantees for related settings, no dedicated implementation is available. Substantial methodological and implementation work would therefore still be required for deployment in our framework.

\subsection{Numerical experiments}\label{sec-simulation_tv}

We next evaluate the estimators in a longitudinal setting with time-dependent confounding affected by prior treatment. As in the baseline simulation, we focus on the causal contrast between the static duration strategies $(1,1,1,0,0)$ and $(1,1,1,1,1)$, and assess estimation of the corresponding RMST difference up to $\tau=10$.

\paragraph{Data-generating process.}
Data are generated for $n$ individuals followed over a discrete grid with $K=5$ visits indexed by $k=0,\dots,4$. Within each interval $(k-1,k]$, event times are generated on a continuous scale under a piecewise-constant hazard model. $X_1$ and $X_2$ denote baseline confounders, whereas $X_3$ and $X_4$ are time-varying covariates, with $X_3$ also acting as a time-varying confounder. The full data-generating process is described in Table~\ref{tab:dgp-timedep}.

As in Section~\ref{sec-simulation_baseline}, time-dependent parameters $\gamma(k)$, $\alpha(k)$, and $\beta(k)$ were introduced to promote positivity throughout follow-up. Four scenarios ranging from low to high confounding and from low to highly informative natural censoring are considered (Appendix~\ref{appendix_chp4}, Section~\ref{sec-coef_tv}).

\begin{table}[H]
\centering
\caption{Data-generating process with time-dependent confounding used in the simulation study.}
\label{tab:dgp-timedep}
\renewcommand{\arraystretch}{1.15}
\setlength{\tabcolsep}{5pt}
\small
\begin{tabular}{p{2.5cm} p{12.7cm}}
\toprule
\textbf{Scenario} & \textbf{Data-generating mechanism / parameters} \\
\midrule

\textbf{General DGP}
&
\(
X_1 \sim \mathcal N(0.5,1),\qquad
X_2 \sim \mathcal N(1.0,1),\qquad
X_{3,0}\sim \mathcal N(1,1),\qquad
X_{4,0}\sim \mathcal N(-1,1)
\)
\\[5pt]

&
for \(k=0,\dots,K-2\), for individuals still at risk at the end of interval \((k,k+1]\),
\\[-1pt]

&
\(
X_{3,k+1}
=
\beta_{\mathrm{X3,intercept}}
+
\beta_{\mathrm{X3,prev}}\,X_{3,k}
+
\beta_{\mathrm{X3,Aprev}}\,A_k
+
\varepsilon_{k+1},
\qquad
\varepsilon_{k+1}\sim\mathcal N(0,\sigma_{\mathrm{X3}}^2)
\)
\\[5pt]

&
\(
X_{4,k+1}
=
\epsilon_{\mathrm{X4,intercept}}
+
\epsilon_{\mathrm{X4,prev}}\,X_{4,k}
+
\eta_{k+1},
\qquad
\eta_{k+1}\sim\mathcal N(0,\sigma_{\mathrm{X4}}^2)
\)
\\[5pt]

&
\(
A_0=1,\quad
\mathbb P(A_{k+1}=1 \mid A_k,X_1,X_2,X_{3,k+1})
=
\operatorname{logit}^{-1}\!\Big(
\gamma_{0,\mathrm{long}}
+\gamma_{X1}X_1
+\gamma_{X2}X_2
+\gamma_{X3}(0.1\,k)X_{3,k+1}
+\gamma_{A\mathrm{prev}}A_k
+\gamma_t(k)
\Big)
\)
\\[3pt]

&
\(
\gamma_t(k)=
\left\{
\begin{array}{ll}
\gamma_{t,\mathrm{pre}}, & k \le 2,\\
\gamma_{t,\mathrm{post}}, & k > 2
\end{array}
\right.
\)
\\[5pt]

&
\(
\lambda_T(k\mid A_k,X_1,X_2,X_{3,k},X_{4,k})
=
\exp\!\Big(
\alpha_0
+\alpha_A A_k
+\alpha_{X1}X_1
+\alpha_{X2}X_2
+\alpha_{X3}X_{3,k}
+\alpha_{X4}X_{4,k}
+\alpha_t(k)
\Big)
\)
\\[3pt]

&
\(
\alpha_t(k)=
\left\{
\begin{array}{ll}
\alpha_{t,\mathrm{pre}}, & k \le 2,\\
\alpha_{t,\mathrm{post}}, & k > 2
\end{array}
\right.
\)
\\[5pt]

&
\(
\lambda_C(k\mid A_k,X_1,X_2,X_{4,k})
=
\beta\,
\exp\!\Big(
\beta_0
+\beta_A A_k
+\beta_{X1}X_1
+\beta_{X2}X_2
+\beta_{X4}X_{4,k}
+\beta_t(k)
\Big)
\)
\\[3pt]

&
\(
\beta_t(k)=
\left\{
\begin{array}{ll}
\beta_{t,\mathrm{pre}}, & k \le 3,\\
\beta_{t,\mathrm{post}}, & k > 3
\end{array}
\right.
\)
\\[5pt]

&
Event and censoring times are generated on the continuous scale within each interval,
recorded as \(T\in(0,\infty)\) and \(C\in(0,\infty)\), with observed follow-up
\(T_{\mathrm{obs}}=\min(T,C)\), administratively truncated at \(\tau\).
\\

\bottomrule
\end{tabular}
\end{table}
From the simulated longitudinal treatment trajectories, we identify individuals naturally following either $(1,1,1,0,0)$ or $(1,1,1,1,1)$. Oracle survival curves under both strategies are additionally computed from the known data-generating mechanism by integrating over the stochastic evolution of the time-dependent confounder $X_3$. These oracle quantities provide the ground truth RMST values against which estimators are evaluated.

\paragraph{Methods.}

We considered a broad range of estimators. Before cloning, we applied a naive difference--in--mean estimator and a G-formula based on observed effective treatment duration. After cloning, we considered a standard Kaplan--Meier estimator, an IPCW Kaplan--Meier accounting only for artificial censoring using a pooled logistic model, a G-formula based on predicted covariate trajectories (see pseudo-code and toy example in Appendix~\ref{sec-implementation} and \ref{sec-toy_example}), and an IPCW Kaplan--Meier separately adjusting for natural censoring with a pooled logistic model and artificial censoring with a piecewise exponential model estimated through Poisson regression.

To evaluate model misspecification, we also considered a single regression model, either pooled logistic or piecewise exponential, for the overall censoring process. In this longitudinal setting, both treatment and outcome models depend on the time-varying confounder $X_{3,k}$. We additionally assessed the impact of ignoring the time-dependent structure.

To illustrate implementation, see the pseudo-code in Section~\ref{sec-implementation} and toy example in Section~\ref{sec-toy_example}, Appendix~\ref{appendix_chp4}. \footnote{To our knowledge, no available R package directly accommodates the full combination of features considered here, namely cloned treatment duration strategies, artificial censoring due to deviation from the assigned strategy, natural censoring, longitudinal covariates, and RMST estimation. }

Results are summarized using boxplots across sample sizes ($n=500$, $1000$, $2000$, $4000$, $8000$), each based on \(50\) independent replicates. Appendix Table~\ref{tab:rmse_tv} reports the bias, ESE, MSE, and RMSE in Scenario~4 (strong confounding and dependent censoring).

\begin{figure}[H]
    \centering
    \includegraphics[width=1\linewidth]{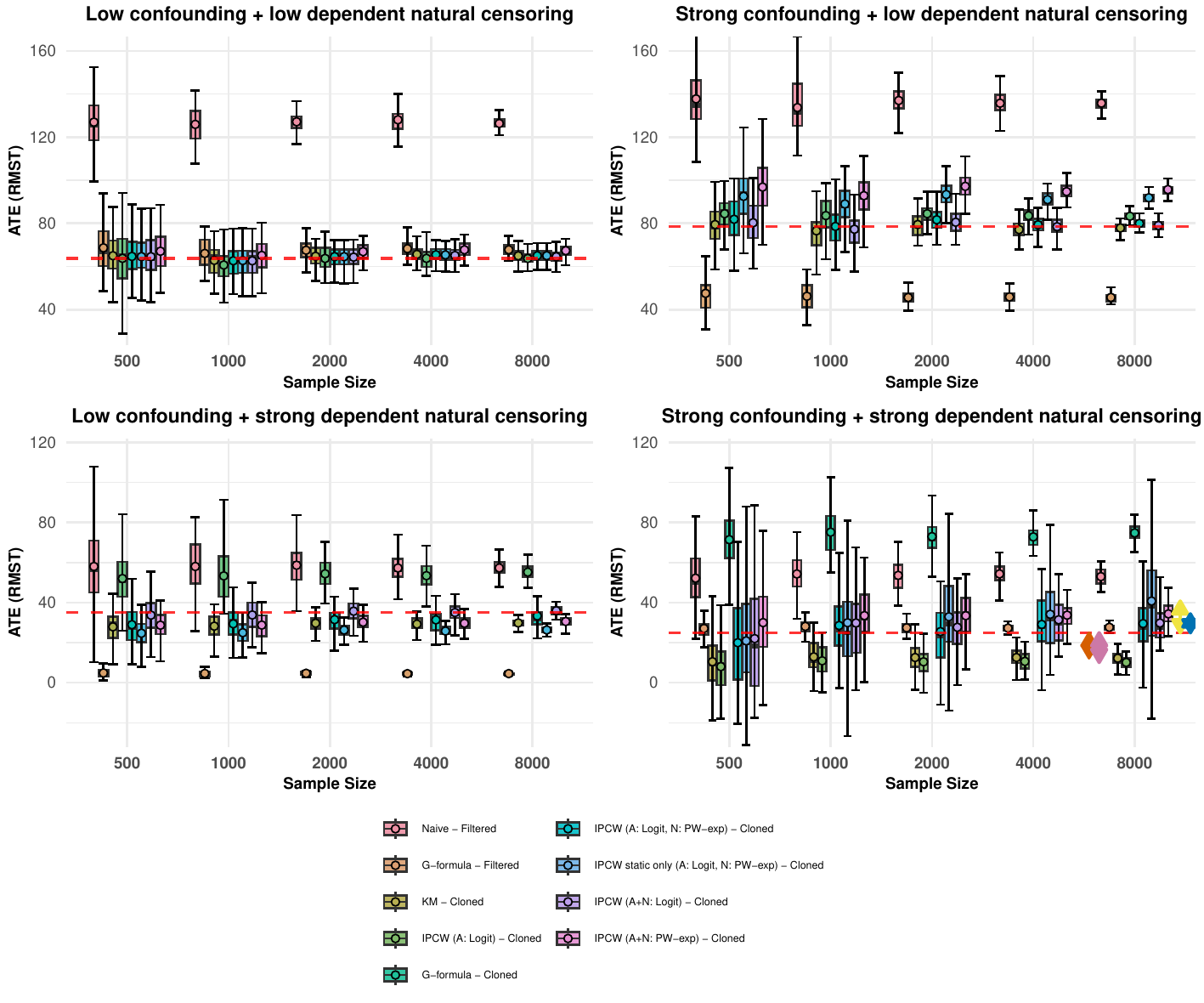}
    \caption{RMST difference estimates after cloning and artificial censoring under four simulation settings with time-dependent covariates, defined by combinations of low or high confounding and low or high dependent natural censoring. In estimator labels, A and N indicate the models used for artificial and natural censoring, respectively.}
    \label{fig:ate_tv}
\end{figure}
As in the previous settings, all estimators applied before cloning are biased because they cannot account for immortal time bias, as illustrated by "G-formula -- Filtered" (\textcolor{red_cbf2}{red arrow}). Similarly, adjusting only for artificial censoring in IPCW leads to substantial bias, whose magnitude depends on the informativeness of natural censoring, as highlighted by the \textcolor{purple_cbf2}{purple arrow} (IPCW (A: Logit) - Cloned).

Modeling the two censoring processes with a single pooled logistic model (IPCW (A+N: Logit) - Cloned) remains competitive, often adequately capturing the overall probability of remaining uncensored even when the two mechanisms arise from different models. Although relying on a single model is risky from a modeling perspective, it may reduce variance when the approximation is adequate. By contrast, the exponential specification (IPCW (A+N: PW--exp) - Cloned) is systematically misspecified and exhibits substantial bias, illustrated by the \textcolor{blue_cbf2}{blue arrow}. IPCW adjustment based only on baseline confounders (IPCW static only (A: Logit, N:PW--exp) - Cloned) is also biased, corresponding to the \textcolor{yellow_cbf2}{yellow arrow}. The consistent estimator in these simulations is IPCW adjusting for artificial censoring with a pooled logistic model and natural censoring with a piecewise exponential model (IPCW (A: Logit, N: PW--exp) - Cloned).

The main difference from the baseline-confounding scenarios is the marked bias of the G-formula in more challenging settings (G-formula - Cloned, right left panel), especially near positivity violations. Because the estimator must predict the full covariate trajectory under each treatment strategy, sequential prediction errors accumulate over time and induce increasing bias as the data-generating mechanism becomes more complex. These simulations also highlight the impact of IPCW misspecification. In complex time-varying settings, IPCW methods tend to outperform G-formula approaches because they rely only on observed covariate trajectories rather than simulated counterfactual ones.

Overall, when only baseline confounders are considered, many approaches are available, including weighting-based, outcome-regression, and doubly robust methods. With time-varying confounders, however, IPCW-based approaches appear to provide the best compromise between implementation complexity and empirical performance, although they remain sensitive to censoring-model misspecification. In the next section, we apply these methods to real-world breast cancer data to evaluate treatment duration effects.

\section{Effect of tamoxifen duration in breast cancer cohort}\label{sec-application}

\section{Conclusion}

In this work, we developed a methodological framework for evaluating treatment duration strategies in observational data. This included formalizing the causal assumptions underlying the cloning, censoring, and weighting (CCW) approach and comparing several estimation strategies through simulation studies. We then applied this framework to the Breast Cancer cohort to investigate the effect of extending treatment duration from 2 to 5 years. 


Our results illustrate the methodological gap between baseline and time-varying settings. Although baseline analyses are simpler, failing to account for time-varying confounding may induce substantial bias. Conversely, incorporating time-dependent covariates increases model complexity, variability, and reliance on stronger modeling assumptions.

Another critical aspect concerns the assumptions underlying the CCW framework, especially consistency under admissible treatment strategies. Assuming that relaxing the treatment rule does not alter the target causal estimand is strong and difficult to verify in practice. Violations may lead to targeting a different causal quantity than intended.

From a practical perspective, even rich cohort data raise substantial challenges, including complex preprocessing, imprecise event timing, and missing data. Extending the methods considered here to more realistic settings, particularly with missing data, remains an important area for future research.

Finally, practitioners face many modeling choices, including missing-data handling, estimator selection, and nuisance-model specification. No single approach can currently be uniformly recommended. Instead, our results support a sensitivity-analysis perspective, where conclusions rely on consistency across plausible specifications.

Overall, this work highlights both the potential and the limitations of causal inference methods for treatment duration in observational data and underscores the need for further methodological developments addressing positivity violations, model uncertainty, and complex longitudinal data structures.

\newpage

\subsubsection*{Acknowledgments}

This research was supported by Sanofi and Inria through the CIFRE program (Convention Industrielle de Formation par la Recherche).


\vspace*{1pc}
\subsubsection*{Conflict of Interest}
Charlotte Voinot and Bernard Sebastien are Sanofi employees and may hold shares and/or stock options in the company.  

\newpage

\bibliographystyle{apalike}
\bibliography{bibliography}
 \newpage
\vspace{2cm}
\begin{center}
    {\Large \textbf{Appendix}}
\end{center}

\begin{appendix}\label{appendix_chp4}

\begin{proposition}[Consistency of the IPCW--Kaplan--Meier estimator with known censoring weights]
\protect\hypertarget{prp-ipcwkm_multiply}{}\label{prp-ipcwkm_multiply}

Under Proposition~\ref{prp:seq-cens-riskset-timedep} and Assumptions~\ref{ass:cens-cond-timedep} (\nameref{ass:cens-cond-timedep}), \ref{ass:cens-cond-pos-timedep} (\nameref{ass:cens-cond-pos-timedep}),
for all $t \in [0,\tau]$, the oracle IPCW--Kaplan--Meier estimator
$S^{*,(d)}_{\mathrm{IPCW}}(t)$ defined in Definition~\ref{def-ipcwkm-chp4_timedep}, using the true censoring weights for both censoring process $W^{(d)}(t)$ defined in Equation~\ref{eq:ipcw-tot-timedep}, is a strongly consistent and asymptotically normal estimator of the counterfactual survival function $S^{(d)}(t)$.
\end{proposition}

\begin{proof}
For $k\le K$,
\begin{align*}
\mathbb P(T^{(d)}\ge k)
&=
\mathbb E\!\left[
\frac{
\mathbb E [\mathbb I(\mathcal N_{k},\ \mathcal E_{k}^{(d) }\mid X^{(d)}_{c},X)
]}{
\mathbb P(\mathcal N_{k},\ \mathcal E_{k}^{(d)}\mid X^{(d)}_{c},X)
}
\,
\mathbb I(T^{(d)}\ge k)
\right] \\[0.2cm]
&=
\mathbb E\!\left[
\frac{
\mathbb I(\mathcal N_{k},\ \mathcal E_{k}^{(d)})
}{
\mathbb P(\mathcal N_{k},\ \mathcal E_{k}^{(d)}\mid X^{(d)}_{c},X)
}
\,
\mathbb I(T^{(d)}\ge k)
\right] \\[0.2cm]
&=
\mathbb E\!\left[
\frac{
1
}{
\mathbb P(\mathcal N_{k},\ \mathcal E_{k}^{(d)}\mid X^{(d)}_{c},X)
}
\,
\mathbb I(T^{(d)} \wedge C \wedge G^{(d)} \geq t_j)
\right] \\[0.2cm]
&=
\mathbb E\!\left[
\frac{1}{
\mathbb P(\mathcal N_{k},\ \mathcal E_{k}^{(d)}\mid X^{(d)}_{c},X)
}
\,
\mathbb I(\widetilde T^{(d)}\ge k)
\right],
\end{align*}

Similarly,
\begin{align*}
\mathbb P(T^{(s)}=k)
&=
\mathbb E\!\left[
\frac{
\mathbb E [\mathbb I(\mathcal N_{k},\ \mathcal E_{k}^{(d) }\mid X^{(d)}_{c},X)
]}{
\mathbb P(\mathcal N_{k},\ \mathcal E_{k}^{(d)}\mid X^{(d)}_{c},X)
}
\,
\mathbb I(T^{(d)}=k)
\right] \\[0.2cm]
&=
\mathbb E\!\left[
\frac{
\mathbb I(\mathcal N_{k},\ \mathcal E_{k}^{(d)})\mathbb I(T^{(d)}=k)
}{
\mathbb P(\mathcal N_{k},\ \mathcal E_{k}^{(d)}\mid X^{(d)}_{c},X)
}
\right] \\[0.2cm]
&=
\mathbb E\!\left[
\frac{\Delta^{(d)}_{k}}{
\mathbb P(\mathcal N_{k},\ \mathcal E_{k}^{(d)}\mid X^{(d)}_{c},X)
}
\,
\mathbb I(\tilde{T}^{(d)}=k)
\right],
\end{align*}
where $\Delta^{(d)}_{k}$ is the event indicator at time $k$ in the cloned
dataset.

Therefore,
\[
\widehat S^{(d)}_{\mathrm{IPCW}}(k)
\;\xrightarrow{\text{a.s.}}\;
\prod_{j\le k}
\left(
1-\frac{\mathbb P(T^{(d}=j)}{\mathbb P(T^{(d)}\ge j)}
\right)
=
S^{(d)}(k),
\]
which proves strong consistency. Asymptotic normality follows from
standard martingale arguments for IPCW Kaplan--Meier estimators.
\end{proof}
\newpage

\section{Implementation}\label{sec-implementation}

The package \href{https://genentech.github.io/survivalCCW/index.html}{\texttt{survivalCCW}} \citep{survivalCCW} provides a useful reference implementation for clone--censor--weight analyses in survival settings, but it is intentionally a light-weight tool designed to streamline a fairly standard CCW workflow rather than a flexible framework for more general longitudinal settings. In particular, it does not directly accommodate our setup with time-varying covariates, more complex censoring structures or natural informative censoring, and its weighting procedure is geared toward the package's built-in CCW implementation rather than user-defined longitudinal extensions. The package \href{https://cran.r-project.org/package=gfoRmula}{\texttt{gfoRmula}} \citep{McGrath2020}, in turn, implements the classical parametric G-formula for longitudinal interventions in the standard longitudinal setting, rather than within a cloning--censoring framework. It should therefore be viewed here as a methodological benchmark rather than as a direct implementation of our procedure. In particular, its use after cloning with artificial censoring is not straightforward in its current form, as such an adaptation would require a non-trivial reformulation of the observed data structure, the censoring process, and the target estimand. This is especially true in time-dependent settings where censoring may itself depend on the evolving covariate and treatment history. 

\newpage
\newgeometry{top=0.5cm,bottom=0.5cm,left=0.5cm,right=0.5cm}
\SetAlFnt{\footnotesize}
\SetKwInOut{KwIn}{Input}
\SetKwInOut{KwOut}{Output}

\refstepcounter{algorithmct}
\begin{tcolorbox}[
  breakable,
  enhanced,
  colback=white,
  colframe=black!15,
  boxrule=0.8pt,
  sharp corners,
  title={Algorithm \thealgorithmct: IPCW estimation for artificial and natural censoring (no time interactions or stabilization): example},
  fonttitle=\bfseries,
  coltitle=black
]
\label{alg:ipcw_artificial_natural}
\DontPrintSemicolon

\KwIn{
  Cloned dataset \texttt{data} (\(n=2676\) clones from \(n=1338\) individuals) with variables: \texttt{ID\_clone}, \texttt{A} (binary strategy), \texttt{T\_obs}, \texttt{status}, \texttt{art.censor}; \\
  Visit grids \texttt{visits1} and \texttt{visits0}, defining the visit-based intervals in each arm; \\
  Covariates for the artificial censoring model: \texttt{X\_art}; Covariates for the natural censoring model: \texttt{X\_nat}
}
\KwOut{Weighted KM, RMST estimates by arm, RMST difference}

\textbf{Step 1: Construct the visit-based long-format datasets} \;

\ForEach{arm \(a \in \{0,1\}\)}{
  Restrict to clones with \(A=a\)\;

  Let \(0=v_0^{(a)} < v_1^{(a)} < \cdots < v_{K_a}^{(a)}\) denote the visit times in \texttt{visits$a$}\;

  Split the interval \([0,\min(T_{\mathrm{obs}},\tau)]\) at each visit time in \texttt{visits$a$} using \texttt{survSplit} (package \texttt{\href{https://cran.r-project.org/web/packages/survival/index.html}{survival}} \citep{Therneau2001}). This defines the visit-based intervals \([v_{j-1}^{(a)}, v_j^{(a)})\) on which both censoring models are specified\;

  For each interval, define: the event indicator, the artificial censoring indicator, the natural censoring indicator, the visit index, the covariates : the time-varying covariates using the 6-month look-back window before each visit and also baseline covariates\;

  Merge the resulting long-format datasets by \((\texttt{ID\_clone}, T_{\mathrm{start}}, T_{\mathrm{stop}})\)\;
}

Let \(\mathcal D_{\mathrm{visit}}^{(1)}\) and \(\mathcal D_{\mathrm{visit}}^{(0)}\) denote the visit-based long-format datasets in arms \(a=1\) and \(a=0\), respectively.

\textbf{Step 2: Estimate artificial censoring weights on the visit-based grid} \;

\ForEach{arm-specific dataset \(\mathcal D_{\mathrm{visit}}^{(a)}\)}{

  Fit a pooled logistic regression model:
  \[
    \Pr\!\bigl(\mathrm{art.censor}_{ij}=1 \mid \bar H_{ij}\bigr)
    =
    \mathrm{logit}^{-1}
    \Bigl(
    \alpha_j + \theta^\top X^{\mathrm{art}}_{ij}
    \Bigr),
  \]
  where \(\alpha_j\) is a visit-specific intercept, and \(X^{\mathrm{art}}_{ij} = (X_i, X_{ij})\) collects baseline and time-varying covariates for individual \(i\) in visit interval \(j\)\;

  For each individual \(i\) and visit interval \(j\), compute
  \[
    \hat p^{\mathrm{art}}_{ij}
    =
    \Pr\!\bigl(\mathrm{art.censor}_{ij}=1 \mid \bar H_{ij}\bigr),
    \qquad
    \hat p^{\mathrm{stay,art}}_{ij}
    =
    1-\hat p^{\mathrm{art}}_{ij}.
  \]

  Compute the estimated probability of remaining artificially uncensored up to the end of interval \(j\):
  \[
    \hat G_i^{\mathrm{art}}(v_j^{(a)})
    =
    \prod_{k \le j} \hat p^{\mathrm{stay,art}}_{ik}.
  \]

  Define the artificial censoring weight:
  \[
    \hat w_i^{\mathrm{art}}(v_j^{(a)})
    =
    \frac{1}{\hat G_i^{\mathrm{art}}(v_j^{(a)})}.
  \]
}

Because artificial censoring is only assessed at visit times, \(\hat w_i^{\mathrm{art}}(t)\) is constant within each visit interval\;

\textbf{Step 3: Estimate natural censoring weights on the visit-based grid} \;

\ForEach{arm-specific dataset \(\mathcal D_{\mathrm{visit}}^{(a)}\)}{

  Fit a piecewise exponential model for natural censoring by first fitting interval-specific hazards:
  \[
    \lambda_{ij}^{\mathrm{nat}}
    =
    \exp\!\Bigl(
    \alpha_j + \theta^\top X^{\mathrm{nat}}_{ij}
    \Bigr),
  \]
  where \(j\) indexes the visit interval \([v_{j-1}^{(a)}, v_j^{(a)})\), and \(X^{\mathrm{nat}}_{ij} = (X_i, X_{ij})\), so that the natural censoring hazard is assumed constant within each visit interval\;

  For each complete visit interval \(j\), define the estimated probability of remaining naturally uncensored over the whole interval:
  \[
    \hat p^{\mathrm{stay,nat}}_{ij}
    =
    \exp\!\Bigl(
      -\hat\lambda_{ij}^{\mathrm{nat}}
      \bigl(v_j^{(a)}-v_{j-1}^{(a)}\bigr)
    \Bigr).
  \]

  More generally, for a time \(t \in [v_{j-1}^{(a)}, v_j^{(a)})\), define the estimated probability of remaining naturally uncensored from the beginning of interval \(j\) up to time \(t\):
  \[
    \hat p^{\mathrm{stay,nat}}_{ij}(t)
    =
    \exp\!\Bigl(
      -\hat\lambda_{ij}^{\mathrm{nat}}
      \bigl(t-v_{j-1}^{(a)}\bigr)
    \Bigr).
  \]
}

\textbf{Step 4: Refine the data at observed event times for weight evaluation} \;

\ForEach{arm \(a \in \{0,1\}\)}{
  Starting from the visit-based dataset \(\mathcal D_{\mathrm{visit}}^{(a)}\), define a refined dataset \(\mathcal D_{\mathrm{eval}}^{(a)}\) by introducing the observed event times in arm \(a\) as additional cut points for weight evaluation\;
}

This additional splitting is used only to evaluate the censoring weights at the exact event times needed for the weighted Kaplan--Meier estimator. It does \emph{not} redefine the censoring models: the pooled logistic model and the piecewise exponential model remain indexed by the original visit intervals through the visit index\;

\textbf{Step 5: Evaluate the censoring weights at each event time} \;

\ForEach{individual \(i\) and refined interval ending at event time \(t\)}{

  Let \(j\) denote the original visit interval such that \(t \in [v_{j-1}^{(a)}, v_j^{(a)})\)\;

  The artificial censoring weight at time \(t\) is
  \[
    \hat w_i^{\mathrm{art}}(t)
    =
    \hat w_i^{\mathrm{art}}(v_j^{(a)}),
  \]
  that is, the value attached to the corresponding visit interval\;

  The natural censoring survival probability at time \(t\) is obtained by combining the contributions from all completed visit intervals before \(t\) with the partial contribution from the current visit interval:
  \[
    \hat G_i^{\mathrm{nat}}(t)
    =
    \prod_{\ell < j} \hat p^{\mathrm{stay,nat}}_{i\ell}
    \times
    \hat p^{\mathrm{stay,nat}}_{ij}(t).
  \]

  Define the natural censoring weight:
  \[
    \hat w_i^{\mathrm{nat}}(t)
    =
    \frac{1}{\hat G_i^{\mathrm{nat}}(t)}.
  \]

  Define the final IPCW weight:
  \[
    \hat w_i(t)
    =
    \hat w_i^{\mathrm{art}}(t)\times \hat w_i^{\mathrm{nat}}(t).
  \]
}

\textbf{Step 6: Weighted survival analysis} \;

Stack both refined arm-specific datasets. For each arm \(a \in \{0,1\}\), let
\(t_{(1)}<\dots<t_{(m)}\) denote the ordered distinct event times observed in arm \(a\)\;

At each event time \(t_k\), compute the weighted number of events and the weighted risk set size:
\[
  d_{a,k}
  =
  \sum_{i:\,A_i=a} \mathbf{1}\!\left(T_{i,\mathrm{stop}}=t_k,\ \delta_i=1\right)\hat w_i(t_k),
  \qquad
  Y_{a,k}
  =
  \sum_{i:\,A_i=a} \mathbf{1}\!\left(T_{i,\mathrm{start}}<t_k \le T_{i,\mathrm{stop}}\right)\hat w_i(t_k),
\]

The weighted Kaplan--Meier estimator in arm \(a\) is
\[
  \widehat S^{(a)}(t)
  =
  \prod_{t_{(k)}\le t}
  \left(1-\frac{d_{a,k}}{Y_{a,k}}\right).
\]

Compute the restricted mean survival time up to \(\tau\):
\[
  \mathrm{RMST}^{(a)}
  =
  \int_0^\tau \widehat S^{(a)}(t)\,\mathrm dt.
\]

\Return{RMST difference \(\mathrm{RMST}^{(1)}-\mathrm{RMST}^{(0)}\), arm-specific RMSTs, and the weighted long-format dataset.}

\end{tcolorbox}
\setcounter{AlgoLine}{0}
\refstepcounter{algorithmct}
\begin{tcolorbox}[
  breakable,
  enhanced,
  colback=white,
  colframe=black!15,
  boxrule=0.8pt,
  sharp corners,
  title={Algorithm \thealgorithmct: Parametric g-computation for RMST under static treatment strategies},
  fonttitle=\bfseries,
  coltitle=black
]
\label{alg:gformula_parametric_general}
\DontPrintSemicolon

\KwIn{
Observed longitudinal data; cloned data; target strategies $g_0$ and $g_1$; horizon $\tau$
}
\KwOut{
Estimated RMST under $g_0$ and $g_1$, and RMST difference
}
\textbf{Step 1: Fit pooled transition models for time-varying covariates} \;

For each time-varying covariate:\;

\Indp
- Construct a transition dataset with one row per individual and per transition \(k \to k+1\), \(k=0,\dots,K-2\)\;

- Fit a pooled parametric transition model of \(X_{k+1}\) given the observed history at time \(k\)\;

- Store the fitted model and residual standard deviation\;

\Indm

\textbf{Step 2: Fit strategy-specific outcome models on the observed data} \;

Within each strategy arm, fit a Cox proportional hazards model in counting-process form, using the observed longitudinal covariate history as predictors. The model includes baseline covariates and time-varying covariates updated at visit times, together with a baseline hazard stratified by visit interval:
\[
\lambda_a(t \mid X_i, \bar X_i(t))
=
\lambda_{0,a}^{(j)}(t)
\exp\!\Bigl(
\beta_a^\top X_i
+
\gamma_a^\top \bar X_i(t)
\Bigr),
\qquad t \in (v_{j-1},v_j].
\]
This specification allows the baseline hazard to vary flexibly across visit-defined intervals while accounting for time-varying covariates.

\textbf{Step 3: Reconstruct strategy-specific covariate trajectories} \;

For each target strategy $g_a$, the post-baseline covariate history is reconstructed sequentially from the fitted transition model through an averaged stochastic approximation of the next-step prediction.
For each target strategy $g_a$, we reconstructed the covariate history that would arise under $g_a$ for each individual.\;

\textbf{Step 4: Predict counterfactual survival under each strategy} \;

For each individual and each strategy $g_a$, plug the reconstructed covariate trajectory under $g_a$ into the fitted outcome model to predict the corresponding survival curve
\[
\widehat S_i^{(a)}(t), \qquad 0 \le t \le \tau.
\]

\textbf{Step 5: Compute individual and marginal RMST} \;

For each individual and strategy $g_a$, compute
\[
\widehat{\mathrm{RMST}}_i^{(a)}
=
\int_0^\tau \widehat S_i^{(a)}(t)\,dt.
\]

Estimate the marginal RMST under strategy $g_a$ by empirical averaging:
\[
\widehat{\mathrm{RMST}}^{(a)}
=
\frac{1}{n}\sum_{i=1}^n \widehat{\mathrm{RMST}}_i^{(a)}.
\]

\textbf{Step 6: Compute the contrast of interest} \;

Compute
\[
\widehat\theta_{\mathrm{RMST}}
=
\widehat{\mathrm{RMST}}^{(1)}-\widehat{\mathrm{RMST}}^{(0)}.
\]

\Return{$\widehat{\mathrm{RMST}}^{(0)}$, $\widehat{\mathrm{RMST}}^{(1)}$, and $\widehat\theta_{\mathrm{RMST}}$.}

\end{tcolorbox}

\restoregeometry
Unlike the classical parametric G-formula of Robins \citep{robins1986new}, which models the full longitudinal data-generating process under a target intervention (including the joint evolution of time-varying covariates, treatment, censoring, and outcome), our approach focuses on reconstructing unobserved time-varying covariate trajectories under each target strategy and combines them with a semi-parametric outcome model fitted on the observed data to obtain standardized RMST predictions.

\section{Illustrative toy examples with explicit calculations}\label{sec-toy_example}

We now provide two toy examples with explicit calculations to illustrate the main steps of Algorithms~\ref{alg:ipcw_artificial_natural} (only for artificial censoring) and~\ref{alg:gformula_parametric_general}. For simplicity, we consider a setting with three individuals, two post-baseline decision times \(t_1=1\), \(t_2=2\), two time-varying covariates \(X_{3,k}\) (continuous) and \(X_{4,k}\) (binary), and two target static treatment strategies
\[
g_0=(1,0,0),
\qquad
g_1=(1,1,1).
\]
We assume a horizon \(\tau=3\). The observed event time is denoted by \(T_{\mathrm{obs}}\), and \(\Delta=1\) indicates that the event occurred before \(\tau\).
\begin{figure}[H]
\centering
\begin{tikzpicture}[->, >=stealth, thick]
    \def\r{0.95cm}

    \node (L0) [draw, circle, minimum size=\r] at (0,2.8) {$X_{3,0}$};
    \node (L1) [draw, circle, minimum size=\r] at (1.9,2.8) {$X_{3,1}$};
    \node (L2) [draw, circle, minimum size=\r] at (3.8,2.8) {$X_{3,2}$};

    \node (Z0) [draw, circle, minimum size=\r] at (0,1.4) {$X_{4,0}$};
    \node (Z1) [draw, circle, minimum size=\r] at (1.9,1.4) {$X_{4,1}$};
    \node (Z2) [draw, circle, minimum size=\r] at (3.8,1.4) {$X_{4,2}$};

    \node (A0) [draw, circle, minimum size=\r] at (0,0) {$A_0$};
    \node (A1) [draw, circle, minimum size=\r] at (1.9,0) {$A_1$};
    \node (A2) [draw, circle, minimum size=\r] at (3.8,0) {$A_2$};

    \node (Y1) [draw, circle, minimum size=\r] at (0.95,-1.4) {$Y_1$};
    \node (Y2) [draw, circle, minimum size=\r] at (2.85,-1.4) {$Y_2$};
    \node (Y3) [draw, circle, minimum size=\r] at (4.75,-1.4) {$Y_3$};

    \draw (L0) -- (L1);
    \draw (L1) -- (L2);

    \draw (Z0) -- (Z1);
    \draw (Z1) -- (Z2);

    \draw (A0) -- (A1);
    \draw (A1) -- (A2);

    \draw (Y1) -- (Y2);
    \draw (Y2) -- (Y3);

    \draw[red_cbf2, bend right=38] (L0) to (A0);
    \draw[red_cbf2, bend left=8]  (Z0) to (A0);

    \draw[red_cbf2, bend right=38] (L1) to (A1);
    \draw[red_cbf2, bend left=8]  (Z1) to (A1);

    \draw[red_cbf2, bend right=38] (L2) to (A2);
    \draw[red_cbf2, bend left=8]  (Z2) to (A2);

    \draw[red_cbf2, bend left=12] (A0) to (L1);
    \draw[red_cbf2, bend left=20] (A0) to (Z1);

    \draw[red_cbf2, bend left=12] (A1) to (L2);
    \draw[red_cbf2, bend left=20] (A1) to (Z2);

    \draw[red_cbf2, bend left=18] (L0) to (Y1);
    \draw[red_cbf2, bend left=20]  (Z0) to (Y1);

    \draw[red_cbf2, bend left=18] (L1) to (Y2);
    \draw[red_cbf2, bend left=20]  (Z1) to (Y2);

    \draw[red_cbf2, bend left=18] (L2) to (Y3);
    \draw[red_cbf2, bend left=20]  (Z2) to (Y3);

    \draw[purple_cbf2] (A0) -- (Y1);
    \draw[purple_cbf2] (A1) -- (Y2);
    \draw[purple_cbf2] (A2) -- (Y3);

    \draw[purple_cbf2] (Y1) -- (A1);
    \draw[purple_cbf2] (Y2) -- (A2);

\end{tikzpicture}
\caption{Mini longitudinal DAG for the toy example. At each time \(t\), the current covariates \((X_{3,t},X_{4,t})\) affect treatment assignment \(A_t\); treatment affects future covariates; and both covariate and treatment histories affect the subsequent outcome process. $X_1$ and $X_2$, the baseline confounders, are not represented here. }
\label{fig:mini_dag_toy_y}
\end{figure}
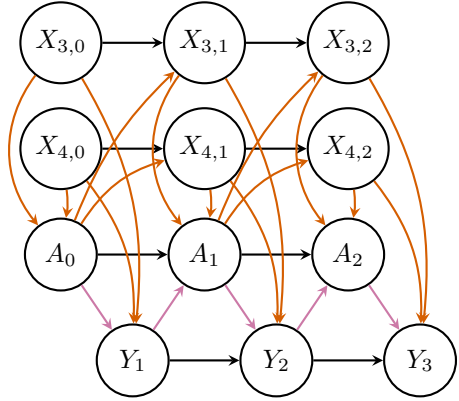

\subsection{Toy example for Algorithm~\ref{alg:ipcw_artificial_natural}: artificial censoring and IPCW}
\paragraph{Observed data.}
Table~\ref{tab:toy_obs_ipcw} displays a small observed dataset with baseline covariates \(X_1,X_2\), treatment history \((A_0,A_1,A_2)\), time-varying covariates \((X_{3,0},X_{3,1},X_{3,2})\) and \((X_{4,0},X_{4,1},X_{4,2})\), and the observed outcome. 

\begin{table}[H]
\centering
\caption{Observed toy data for the IPCW illustration.}
\label{tab:toy_obs_ipcw}
\small
\begin{tabular}{c c c c c c c c c c c c c}
\toprule
ID & \(X_1\) & \(X_2\) & \(A_0\) & \(A_1\) & \(A_2\) & \(X_{3,0}\) & \(X_{3,1}\) & \(X_{3,2}\) & \(X_{4,0}\) & \(X_{4,1}\) & \(X_{4,2}\) & \((T_{\mathrm{obs}},\Delta)\) \\
\midrule
1 & 0 & 1 & 1 & 0 & 0 & 2.0 & 1.8 & NA & 0 & 1 & NA & \((1.5,1)\) \\
2 & 1 & 0 & 1 & 1 & 0 & 1.0 & 1.4 & 1.1 & 1 & 1 & 0 & \((2.5,1)\) \\
3 & 0 & 0 & 1 & 1 & 1 & 0.5 & 0.7 & 0.8 & 0 & 1 & 1 & \((3,0)\) \\
\bottomrule
\end{tabular}
\end{table}

The observed treatment histories are therefore
\[
\bar A_3^{(1)}=(1,0,0),\qquad
\bar A_3^{(2)}=(1,1,0),\qquad
\bar A_3^{(3)}=(1,1,1).
\]

\paragraph{Step 1: cloning and artificial censoring.}
Each individual is cloned into two copies, one assigned to \(g_0=(1,0,0)\) and one assigned to \(g_1=(1,1,1)\). A clone is artificially censored as soon as the observed treatment no longer matches its assigned strategy.

\begin{table}[H]
\centering
\caption{Cloned toy data and artificial censoring times.}
\label{tab:toy_cloned_ipcw}
\small
\setlength{\tabcolsep}{4pt}
\begin{tabular}{c c p{2.2cm} p{1.7cm} p{2.5cm} p{3.7cm}}
\toprule
ID\_clone & ID & Assigned strategy & Compatible up to & Artificial censoring time & Final observed follow-up \\
\midrule
1a & 1 & \(g_0=(1,0,0)\) & \(t_3\) & none & event at 1.5 \\
1b & 1 & \(g_1=(1,1,1)\) & \(t_1\) & \(t_2=2\) & event at 1.5 \\
2a & 2 & \(g_0=(1,0,0)\) & \(t_1\) & \(t_2=2\) & artificially censored at 2 \\
2b & 2 & \(g_1=(1,1,1)\) & \(t_2\) & \(t_3=3\) & event at 2.5 \\
3a & 3 & \(g_0=(1,0,0)\) & \(t_1\) & \(t_2=2\) & artificially censored at 2 \\
3b & 3 & \(g_1=(1,1,1)\) & \(t_3\) & none & naturally censored at 3 \\
\bottomrule
\end{tabular}
\end{table}

Thus:
\begin{itemize}
    \item in arm \(g_0\), clone \(1a\) remains adherent, whereas clones \(2a\) and \(3a\) are artificially censored at \(t_2\);
    \item in arm \(g_1\), clone \(1b\) is artificially censored at \(t_2\), clone \(2b\) is artificially censored at \(t_3\), and clone \(3b\) remains adherent.
\end{itemize}

\paragraph{Step 2: long format within each arm.}
To match Algorithm~\ref{alg:ipcw_artificial_natural}, the follow-up is split at the visit times. For illustration, consider the arm \(g_0\). The corresponding long-format data are shown in Table~\ref{tab:toy_long_g0}. This can be done using \texttt{survSplit} (package \texttt{\href{https://cran.r-project.org/web/packages/survival/index.html}{survival}} \citep{Therneau2001}). 

\begin{table}[H]
\centering
\caption{Long-format data in arm \(g_0\).}
\label{tab:toy_long_g0}
\small
\begin{tabular}{c c c c c c c c}
\toprule
ID\_clone & Interval & \(T_{\mathrm{start}}\) & \(T_{\mathrm{stop}}\) & \(X_{3,k}\) & \(X_{4,k}\) & art.censor & status\\
\midrule
1a & \(k=1\) & 0 & 1   & 2.0 & 0 & 0 & 0 \\
1a & \(k=2\) & 1 & 1.5 & 1.8 & 1 & 0 & 1  \\
2a & \(k=1\) & 0 & 1   & 1.0 & 1 & 0 & 0 \\
2a & \(k=2\) & 1 & 2   & 1.4 & 1 & 1 & 0  \\
3a & \(k=1\) & 0 & 1   & 0.5 & 0 & 0 & 0 \\
3a & \(k=2\) & 1 & 2   & 0.7 & 1 & 1 & 0  \\
\bottomrule
\end{tabular}
\end{table}

There is no interval \(k=3\) in arm \(g_0\), because all remaining follow-up has either ended with an event or been artificially censored by time \(2\).
\paragraph{Step 3: pooled logistic model for artificial censoring.}
Suppose that, in arm \(g_0\), one pooled logistic regression for artificial censoring is
\[
\Pr(\mathrm{art.censor}_k=1\mid X_1,X_2,X_{3,k},X_{4,k},\alpha_j)
=
\mathrm{logit}^{-1}\!\bigl(
-3 + 1.2\,\mathbf 1(k=2) + 0.8X_1 + 0.5X_2 + 0.6X_{3,k} + 1.0X_{4,k}
\bigr).
\]
We now compute the predicted probability of artificial censoring for each interval.

For clone \(1a\):
\[
\hat p_{1a,1}
=
\mathrm{logit}^{-1}(-3+0+0.8\times0+0.5\times1+0.6\times2.0+1.0\times0)
=
\mathrm{logit}^{-1}(-1.3)
\approx 0.214,
\]
\[
\hat p_{1a,2}
=
\mathrm{logit}^{-1}(-3+1.2+0.8\times0+0.5\times1+0.6\times1.8+1.0\times1)
=
\mathrm{logit}^{-1}(0.78)
\approx 0.686.
\]

For clone \(2a\):
\[
\hat p_{2a,1}
=
\mathrm{logit}^{-1}(-3+0+0.8\times1+0.5\times0+0.6\times1.0+1.0\times1)
=
\mathrm{logit}^{-1}(-0.6)
\approx 0.354,
\]
\[
\hat p_{2a,2}
=
\mathrm{logit}^{-1}(-3+1.2+0.8\times1+0.5\times0+0.6\times1.4+1.0\times1)
=
\mathrm{logit}^{-1}(0.84)
\approx 0.698.
\]

For clone \(3a\):
\[
\hat p_{3a,1}
=
\mathrm{logit}^{-1}(-3+0+0.8\times0+0.5\times0+0.6\times0.5+1.0\times0)
=
\mathrm{logit}^{-1}(-2.7)
\approx 0.063,
\]
\[
\hat p_{3a,2}
=
\mathrm{logit}^{-1}(-3+1.2+0.8\times0+0.5\times0+0.6\times0.7+1.0\times1)
=
\mathrm{logit}^{-1}(-0.38)
\approx 0.406.
\]

The corresponding estimated probabilities of \emph{remaining} artificially uncensored are
\[
\hat p^{\mathrm{stay}}_{ij}=1-\hat p_{ij}.
\]

\begin{table}[H]
\centering
\caption{Predicted artificial censoring probabilities and IPCW components in arm \(g_0\).}
\label{tab:toy_probs_g0}
\small
\begin{tabular}{c c c c c}
\toprule
ID\_clone & Interval \(k\) & \(\hat p_{ik}\) & \(\hat p^{\mathrm{stay}}_{ik}=1-\hat p_{ik}\) & \(S_i(t_k)\) \\
\midrule
1a & 1 & 0.214 & 0.786 & 1 \\
1a & 2 & 0.686 & 0.314 & 0.786 \\
2a & 1 & 0.354 & 0.646 & 1 \\
2a & 2 & 0.698 & 0.302 & 0.646 \\
3a & 1 & 0.063 & 0.937 & 1 \\
3a & 2 & 0.406 & 0.594 & 0.937 \\
\bottomrule
\end{tabular}
\end{table}

Recall that
\[
G_i(t_k)=\prod_{j<k}\hat p^{\mathrm{stay}}_{ij}.
\]
Therefore, at the beginning of interval \(k=2\),
\[
G_{1a}(t_2)=0.786,\qquad
G_{2a}(t_2)=0.646,\qquad
G_{3a}(t_2)=0.937.
\]
The corresponding IPCW weights are
\[
w_i(t_k)=\frac{1}{G_i(t_k)}.
\]
Hence,
\[
w_{1a}(t_2)=\frac{1}{0.786}=1.272,\qquad
w_{2a}(t_2)=\frac{1}{0.646}=1.548,\qquad
w_{3a}(t_2)=\frac{1}{0.937}=1.067.
\]

\paragraph{Step 4: Weighted Kaplan--Meier estimator in arm \(g_0\).}
This step and the next one can be implemented directly using the \texttt{survfit} function from the package \texttt{\href{https://cran.r-project.org/web/packages/survival/index.html}{survival}} \citep{Therneau2001}, by supplying the long-format dataset together with the corresponding weights. For clarity, however, we detail the calculation below.

In arm \(g_0\), the only observed event occurs for clone \(1a\) at time \(1.5\). At that time, all three clones are still under observation, since clones \(2a\) and \(3a\) are artificially censored only at time \(2\). Thus,
\[
d_{0}(1.5)=w_{1a}(1.5)=1,
\]
because before time \(1.5\), no artificial censoring has yet occurred, so \(G_i(1.5)=1\) and the weight is 1 for all individuals.

The weighted risk set at time \(1.5\) is
\[
Y_{0}(1.5)=1+1+1=3.
\]
Therefore the weighted Kaplan--Meier jump at time \(1.5\) is
\[
\widehat S^{(0)}(1.5)
=
1-\frac{d_0(1.5)}{Y_0(1.5)}
=
1-\frac{1}{3}
=
\frac{2}{3}.
\]
Since there are no further events in arm \(g_0\), the weighted survival curve is
\[
\widehat S^{(0)}(t)=
\begin{cases}
1, & 0\le t<1.5,\\[0.3em]
2/3, & 1.5\le t\le 3.
\end{cases}
\]
The corresponding RMST is
\[
\mathrm{RMST}^{(0)}
=
\int_0^3 \widehat S^{(0)}(t)\,dt
=
\int_0^{1.5} 1\,dt
+
\int_{1.5}^{3} \frac{2}{3}\,dt
=
1.5 + 1.5\times \frac{2}{3}
=
2.5.
\]

\paragraph{Step 5: weighted Kaplan--Meier in arm \(g_1\).}
In arm \(g_1\), clone \(2b\) has an event at time \(2.5\), while clone \(1b\) is artificially censored at \(2\), and clone \(3b\) remains event-free until \(3\).

Suppose that the corresponding artificial censoring model in arm \(g_1\) yields
\[
G_{2b}(2.5)=0.55,\qquad G_{3b}(2.5)=0.80.
\]
Then
\[
w_{2b}(2.5)=\frac{1}{0.55}=1.818,\qquad
w_{3b}(2.5)=\frac{1}{0.80}=1.25.
\]
At time \(2.5\), clone \(1b\) is no longer in the risk set because it was artificially censored at \(2\). Therefore
\[
d_1(2.5)=1.818,
\qquad
Y_1(2.5)=1.818+1.25=3.068.
\]
Hence
\[
\widehat S^{(1)}(2.5)
=
1-\frac{1.818}{3.068}
\approx 0.407.
\]
Since there is no event before \(2.5\), we have
\[
\widehat S^{(1)}(t)=
\begin{cases}
1, & 0\le t<2.5,\\[0.3em]
0.407, & 2.5\le t\le 3.
\end{cases}
\]
Thus
\[
\mathrm{RMST}^{(1)}
=
\int_0^3 \widehat S^{(1)}(t)\,dt
=
2.5 + 0.5\times 0.407
=
2.7035.
\]

Finally, the RMST contrast is
\[
\mathrm{RMST}^{(1)}-\mathrm{RMST}^{(0)}
=
2.7035-2.5
=
0.2035.
\]

\subsection{Toy example for Algorithm~\ref{alg:gformula_parametric_general}: parametric g-computation after cloning and artificial censoring}

We now illustrate the parametric g-computation approach in the same toy longitudinal setting, while keeping the cloning--censoring logic explicit. In contrast with the IPCW approach, artificial censoring is not handled through inverse probability weighting. Instead, after cloning and artificial censoring, we complete the covariate histories under each strategy by imputing only the covariate values that become missing after deviation from the assigned strategy. The observed strategy-compatible covariate values are kept unchanged.

As in the previous toy example, we consider three individuals, two post-baseline decision times \(t_1=1\) and \(t_2=2\), a horizon \(\tau=3\), two time-varying covariates \(X_{3,k}\) (continuous) and \(X_{4,k}\) (binary), and two target static treatment strategies
\[
g_0=(1,0,0),
\qquad
g_1=(1,1,1).
\]

\paragraph{Observed data.}
We start from the same observed data as in Table~\ref{tab:toy_obs_ipcw}. The observed treatment histories are
\[
\bar A_3^{(1)}=(1,0,0),\qquad
\bar A_3^{(2)}=(1,1,0),\qquad
\bar A_3^{(3)}=(1,1,1).
\]

\paragraph{Step 1: cloning, artificial censoring, and creation of missing covariate values.}
Each individual is cloned into two copies, one assigned to \(g_0\) and one assigned to \(g_1\). As in the IPCW illustration, each clone is artificially censored at the first visit where the observed treatment no longer matches the assigned strategy. Covariate values after that time are then treated as missing for that clone, because they are no longer observed under the assigned strategy.

The resulting cloned data are shown in Table~\ref{tab:toy_cloned_gf}. Observed covariate values are retained as long as the clone remains strategy-compatible, and only post-deviation covariates are set to \texttt{NA}.

\begin{table}[H]
\centering
\caption{Cloned data for the g-computation illustration. Covariate values are kept while the clone remains compatible with its assigned strategy, and set to \texttt{NA} after artificial censoring. The variable \texttt{art.censor} indicates whether the clone is artificially censored before the horizon \(\tau\).}
\label{tab:toy_cloned_gf}
\small
\begin{tabular}{c c c c c c c c c c c}
\toprule
ID\_clone & Assigned strategy & \(X_{3,0}\) & \(X_{4,0}\) & \(X_{3,1}\) & \(X_{4,1}\) & \(X_{3,2}\) & \(X_{4,2}\) & \(T_{\mathrm{obs}}\) & \(\Delta\) & art.censor \\
\midrule
1a & \(g_0=(1,0,0)\) & 2.0 & 0 & 1.8 & 1 & \texttt{NA} & \texttt{NA} & 1.5 & 1 & 0 \\
1b & \(g_1=(1,1,1)\) & 2.0 & 0 & 1.8 & 1 & \texttt{NA} & \texttt{NA} & 1.0 & 0 & 1 \\
2a & \(g_0=(1,0,0)\) & 1.0 & 1 & 1.4 & 1 & \texttt{NA} & \texttt{NA} & 1 & 0 & 1 \\
2b & \(g_1=(1,1,1)\) & 1.0 & 1 & 1.4 & 1 & 1.1 & 0 & 2 & 0 & 1 \\
3a & \(g_0=(1,0,0)\) & 0.5 & 0 & 0.7 & 1 & \texttt{NA} & \texttt{NA} & 1 & 0 & 1 \\
3b & \(g_1=(1,1,1)\) & 0.5 & 0 & 0.7 & 1 & 0.8 & 1 & 3 & 0 & 0 \\
\bottomrule
\end{tabular}
\end{table}
Thus:
\begin{itemize}
    \item in arm \(g_0\), clone \(1a\) contributes observed covariate information up to the event at time \(1.5\), whereas clones \(2a\) and \(3a\) are artificially censored at time \(1\). Their values \((X_{3,1},X_{4,1})\) remain observed, but their subsequent covariates \((X_{3,2},X_{4,2})\) are missing under \(g_0\);
    \item in arm \(g_1\), clone \(1b\) is artificially censored at time \(1\), so \((X_{3,1},X_{4,1})\) are still observed but \((X_{3,2},X_{4,2})\) are missing under \(g_1\); clone \(2b\) remains compatible until time \(2\), so both \((X_{3,1},X_{4,1})\) and \((X_{3,2},X_{4,2})\) are observed, although follow-up is then artificially censored before \(\tau\); clone \(3b\) remains fully compatible and therefore contributes a fully observed covariate history.
\end{itemize}

\paragraph{Step 2: fit arm-specific covariate models on the observed strategy-compatible data.}
We next fit, separately within each arm, parametric models for the time-varying covariates using the values that remain observed under the assigned strategy. In this toy example, all clones have observed values for \((X_{3,1},X_{4,1})\), and the only covariates that must be imputed are the second-visit covariates \((X_{3,2},X_{4,2})\) for clones artificially censored before time \(2\). Accordingly, we only specify models for \(X_{3,2}\) and \(X_{4,2}\).

For illustration, suppose that in arm \(g_0\), the fitted models are
\[
X_{3,2}^{(0)} = 0.3 + 0.6X_{3,1} + 0.2X_{4,1},
\qquad
\Pr\!\bigl(X_{4,2}^{(0)}=1 \mid X_{3,1},X_{4,1}\bigr)
=
\mathrm{logit}^{-1}\!\left(-0.8 + 0.7X_{3,1} + 0.9X_{4,1}\right),
\]
and in arm \(g_1\),
\[
X_{3,2}^{(1)} = 0.2 + 0.5X_{3,1} + 0.3X_{4,1},
\qquad
\Pr\!\bigl(X_{4,2}^{(1)}=1 \mid X_{3,1},X_{4,1}\bigr)
=
\mathrm{logit}^{-1}\!\left(-0.6 + 0.6X_{3,1} + 0.8X_{4,1}\right).
\]

These models are fitted on the portions of cloned follow-up that remain compatible with each strategy. In particular, they are not used to regenerate the whole covariate history from baseline for every clone, but only to fill in the covariate values that are missing because of artificial censoring.

\paragraph{Step 3: impute only the missing covariates under each strategy.}
We now complete the cloned dataset by imputing only the \texttt{NA} covariate values under each assigned strategy, while keeping the observed compatible covariates unchanged.

\medskip
\noindent
\textbf{Arm \(g_0\).}
Under \(g_0\), all three clones have observed values for \((X_{3,1},X_{4,1})\), but clone \(1a\) has missing \((X_{3,2},X_{4,2})\) because the event occurs at time \(1.5\), and clones \(2a\) and \(3a\) have missing \((X_{3,2},X_{4,2})\) because they are artificially censored at time \(1\). We therefore impute \((X_{3,2},X_{4,2})\) for all three clones in arm \(g_0\).

For clone \(1a\), with observed \((X_{3,1},X_{4,1})=(1.8,1)\),
\[
\widehat X_{3,2,1a}^{(0)}
=
0.3 + 0.6(1.8) + 0.2(1)
=
1.58,
\]
\[
\widehat{\Pr}\!\bigl(X_{4,2,1a}^{(0)}=1\bigr)
=
\mathrm{logit}^{-1}\!\bigl(-0.8 + 0.7(1.8) + 0.9(1)\bigr)
=
\mathrm{logit}^{-1}(1.36)
\approx 0.796.
\]
Using a deterministic plug-in prediction for this illustration, we set
\[
\widehat X_{4,2,1a}^{(0)}=1.
\]

For clone \(2a\), with observed \((X_{3,1},X_{4,1})=(1.4,1)\),
\[
\widehat X_{3,2,2a}^{(0)}
=
0.3 + 0.6(1.4) + 0.2(1)
=
1.34,
\]
\[
\widehat{\Pr}\!\bigl(X_{4,2,2a}^{(0)}=1\bigr)
=
\mathrm{logit}^{-1}\!\bigl(-0.8 + 0.7(1.4) + 0.9(1)\bigr)
=
\mathrm{logit}^{-1}(1.08)
\approx 0.746.
\]
Using a deterministic plug-in prediction for this illustration, we set
\[
\widehat X_{4,2,2a}^{(0)}=1.
\]

For clone \(3a\), with observed \((X_{3,1},X_{4,1})=(0.7,1)\),
\[
\widehat X_{3,2,3a}^{(0)}
=
0.3 + 0.6(0.7) + 0.2(1)
=
0.92,
\]
\[
\widehat{\Pr}\!\bigl(X_{4,2,3a}^{(0)}=1\bigr)
=
\mathrm{logit}^{-1}\!\bigl(-0.8 + 0.7(0.7) + 0.9(1)\bigr)
=
\mathrm{logit}^{-1}(0.59)
\approx 0.643.
\]
Again using a deterministic plug-in rule, we set
\[
\widehat X_{4,2,3a}^{(0)}=1.
\]

\medskip
\noindent
\textbf{Arm \(g_1\).}
Under \(g_1\), clone \(1b\) is artificially censored at time \(1\), so \((X_{3,1},X_{4,1})\) remain observed but \((X_{3,2},X_{4,2})\) must be imputed. Clones \(2b\) and \(3b\) already have observed values for \((X_{3,2},X_{4,2})\), so no covariate imputation is needed for these clones.

For clone \(1b\), with observed \((X_{3,1},X_{4,1})=(1.8,1)\),
\[
\widehat X_{3,2,1b}^{(1)}
=
0.2 + 0.5(1.8) + 0.3(1)
=
1.4,
\]
\[
\widehat{\Pr}\!\bigl(X_{4,2,1b}^{(1)}=1\bigr)
=
\mathrm{logit}^{-1}\!\bigl(-0.6 + 0.6(1.8) + 0.8(1)\bigr)
=
\mathrm{logit}^{-1}(1.28)
\approx 0.782.
\]
Using a deterministic plug-in prediction for this illustration, we set
\[
\widehat X_{4,2,1b}^{(1)}=1.
\]

The completed covariate histories are summarized in Table~\ref{tab:toy_completed_gf}.

\begin{table}[H]
\centering
\caption{Completed covariate histories after arm-specific imputation of missing covariates only. Observed strategy-compatible values are kept unchanged; only \texttt{NA} values are imputed.}
\label{tab:toy_completed_gf}
\small
\begin{tabular}{c c c c c c c c}
\toprule
ID\_clone & Assigned strategy & \(X_{3,0}\) & \(X_{4,0}\) & \(X_{3,1}\) & \(X_{4,1}\) & \(X_{3,2}\) & \(X_{4,2}\) \\
\midrule
1a & \(g_0\) & 2.0 & 0 & 1.8 & 1 & 1.58 & 1 \\
2a & \(g_0\) & 1.0 & 1 & 1.4 & 1 & 1.34 & 1 \\
3a & \(g_0\) & 0.5 & 0 & 0.7 & 1 & 0.92 & 1 \\
1b & \(g_1\) & 2.0 & 0 & 1.8 & 1 & 1.40 & 1 \\
2b & \(g_1\) & 1.0 & 1 & 1.4 & 1 & 1.1 & 0 \\
3b & \(g_1\) & 0.5 & 0 & 0.7 & 1 & 0.8 & 1 \\
\bottomrule
\end{tabular}
\end{table}
\paragraph{Step 4: fit an arm-specific time-varying Cox model and predict survival under each strategy.}
We now fit, within each arm, a Cox model with time-varying covariates using the cloned data after artificial censoring and covariate completion. Because all missing covariate values have been imputed under the assigned strategy, each clone is associated with a completed covariate history under that strategy.

For illustration, suppose that in each arm \(a\in\{0,1\}\), the hazard is modeled as
\[
\lambda^{(a)}(t \mid X_3(t),X_4(t))
=
\lambda_0^{(a)}(t)\,
\exp\!\bigl(
0.5\,X_3(t)+0.6\,X_4(t)
\bigr),
\]
where \(X_3(t)\) and \(X_4(t)\) are piecewise constant over the intervals \([0,1)\), \([1,2)\), and \([2,3]\), taking values \((X_{3,0},X_{4,0})\), \((X_{3,1},X_{4,1})\), and \((X_{3,2},X_{4,2})\), respectively.

Consider individual 2. Under \(g_0\), the completed covariate history is
\[
(X_{3,0},X_{4,0},X_{3,1},X_{4,1},X_{3,2},X_{4,2})=(1.0,1,1.4,1,1.34,1),
\]
whereas under \(g_1\), it is
\[
(X_{3,0},X_{4,0},X_{3,1},X_{4,1},X_{3,2},X_{4,2})=(1.0,1,1.4,1,1.1,0).
\]
These completed trajectories are inserted into the fitted arm-specific Cox models to obtain the corresponding predicted survival curves \(\widehat S_2^{(0)}(t)\) and \(\widehat S_2^{(1)}(t)\), from which the strategy-specific restricted mean survival times are computed:
\[
\widehat{\mathrm{RMST}}_2^{(0)}
=
\int_0^3 \widehat S_2^{(0)}(t)\,dt,
\qquad
\widehat{\mathrm{RMST}}_2^{(1)}
=
\int_0^3 \widehat S_2^{(1)}(t)\,dt.
\]
Because the completed covariate history under \(g_1\) yields a lower predicted hazard over the last interval, the resulting RMST is larger under \(g_1\) than under \(g_0\).

Proceeding similarly for all individuals yields the individual-level predicted RMSTs shown in Table~\ref{tab:toy_rmst_gf}.

\begin{table}[H]
\centering
\caption{Individual predicted RMSTs under each strategy in the g-computation toy example.}
\label{tab:toy_rmst_gf}
\small
\begin{tabular}{c c c c}
\toprule
ID & \(\widehat{\mathrm{RMST}}_i^{(0)}\) & \(\widehat{\mathrm{RMST}}_i^{(1)}\) & Difference \\
\midrule
1 & 2.05 & 2.18 & 0.13 \\
2 & 2.12 & 2.29 & 0.17 \\
3 & 2.36 & 2.47 & 0.11 \\
\bottomrule
\end{tabular}
\end{table}

The marginal RMST estimates are then obtained by averaging over the original individuals:
\[
\widehat{\mathrm{RMST}}^{(0)}
=
\frac{2.05+2.12+2.36}{3}
=
2.18,
\qquad
\widehat{\mathrm{RMST}}^{(1)}
=
\frac{2.18+2.29+2.47}{3}
=
2.31,
\]
so that
\[
\widehat\theta_{\mathrm{RMST}}
=
\widehat{\mathrm{RMST}}^{(1)}-\widehat{\mathrm{RMST}}^{(0)}
=
0.13.
\]

\section{Simulation details of Section~\ref{sec-simulation_baseline}} \label{sec-coef_baseline}


\begin{table}[ht]
\centering
\caption{General DGP parameters.}
\label{tab:dgp-s1}
\renewcommand{\arraystretch}{1.2}
\setlength{\tabcolsep}{6pt}
\small
\begin{tabular}{p{4.2cm} p{11.8cm}}
\toprule
\textbf{Block} & \textbf{Specification / parameters} \\
\midrule
\textbf{General DGP} &
$
X_1 \sim \mathcal N(\mu_1,\sigma_1^2),\;
X_2 \sim \mathcal N(\mu_2,\sigma_2^2),\;
X_3 \sim \mathcal N(\mu_3,\sigma_3^2)
$
\\[4pt]
&
for $k \in (0,1,\dots,K)$,
\\[-2pt]
&
$
A_0=1,\quad
\mathbb P(A_k=1 \mid A_{k-1},X)
=
\operatorname{logit}^{-1}\!\big(
\gamma_0+\gamma_A A_{k-1}+\gamma_{X1}X_1+\gamma_{X2}X_2+\gamma_t(k)
\big),
\quad
\gamma_t(k)=
\begin{cases}
\gamma_{t,\mathrm{pre}}, & k \le 2,\\
\gamma_{t,\mathrm{post}}, & k > 2
\end{cases}
$
\\[6pt]
&
$
\lambda_T(k \mid A_k,X)
=
\exp\!\big(
\alpha_0+\alpha_A A_k+\alpha_{X1}X_1+\alpha_{X2}X_2+\alpha_{X3}X_3+\alpha_t(k)
\big),
\quad
\alpha_t(k)=
\begin{cases}
\alpha_{t,\mathrm{pre}}, & k \le 2,\\
\alpha_{t,\mathrm{post}}, & k > 2
\end{cases}
$
\\[6pt]
&
$
\lambda_C(k \mid A_k,X)
=
\beta\exp\!\big(
\beta_0+\beta_A A_k+\beta_{X2}X_2+\beta_{X3}X_3+\beta_t(k)
\big),
\quad
\beta_t(k)=
\begin{cases}
\beta_{t,\mathrm{pre}}, & k \le 2,\\
\beta_{t,\mathrm{post}}, & k > 2
\end{cases}
$
\\
\midrule

\end{tabular}
\end{table}
\normalsize
\paragraph{Scenario 1 (Low confounding and low dependent censoring).} 
\begin{itemize}
    \item \textbf{Treatment parameters} 
$\gamma_0=1.1,\;\gamma_A=1,\;\gamma_{X1}=0.2,\;\gamma_{X2}=-0.05,\;
\gamma_{t,\mathrm{pre}}=0.4,\;\gamma_{t,\mathrm{post}}=-2.2$
\item \textbf{Survival parameters} 
$\alpha_0=-2.7,\;\alpha_A=-1.0,\;\alpha_{X1}=0.1,\;\alpha_{X2}=-0.08,\;\alpha_{X3}=-0.2,\;
\alpha_{t,\mathrm{pre}}=-0.2,\;\alpha_{t,\mathrm{post}}=0.3$
\item \textbf{Censoring parameters} 
$\beta=0.08,\;\beta_0=-0.4,\;\beta_A=0.05,\;\beta_{X2}=-0.2,\;\beta_{X3}=0.1,\;
\beta_{t,\mathrm{pre}}=0,\;\beta_{t,\mathrm{post}}=0$
\end{itemize}

\paragraph{Scenario 2 (High confounding and low dependent censoring).} 
\begin{itemize}
\item \textbf{Treatment parameters} 
$\gamma_0=1.1,\;\gamma_A=1,\;\gamma_{X1}=-0.5,\;\gamma_{X2}=0.5,\;
\gamma_{t,\mathrm{pre}}=0.4,\;\gamma_{t,\mathrm{post}}=-2.2$
\item \textbf{Survival parameters} 
$\alpha_0=-2.7,\;\alpha_A=-1.0,\;\alpha_{X1}=0.5,\;\alpha_{X2}=-0.8,\;\alpha_{X3}=-0.8,\;
\alpha_{t,\mathrm{pre}}=-0.2,\;\alpha_{t,\mathrm{post}}=0.3$
\item \textbf{Censoring parameters} 
$\beta=0.05,\;\beta_0=-0.4,\;\beta_A=0.05,\;\beta_{X2}=-0.2,\;\beta_{X3}=0.1$
\end{itemize}

\paragraph{Scenario 3 (Low confounding and high dependent censoring).} 
\begin{itemize}
\item \textbf{Treatment parameters} 
$\gamma_0=1.1,\;\gamma_A=1,\;\gamma_{X1}=0.2,\;\gamma_{X2}=-0.05,\;
\gamma_{t,\mathrm{pre}}=0.4,\;\gamma_{t,\mathrm{post}}=-2.2$
\item \textbf{Survival parameters} 
$\alpha_0=-1.7,\;\alpha_A=-1.0,\;\alpha_{X1}=0.1,\;\alpha_{X2}=-0.08,\;\alpha_{X3}=-0.2,\;
\alpha_{t,\mathrm{pre}}=-0.2,\;\alpha_{t,\mathrm{post}}=0.3$
\item \textbf{Censoring parameters} 
$\beta=0.1,\;\beta_0=-0.3,\;\beta_A=0.5,\;\beta_{X2}=0.5,\;\beta_{X3}=0.5,\;
\beta_{t,\mathrm{pre}}=-1.5,\;\beta_{t,\mathrm{post}}=0.5$
\end{itemize}

\paragraph{Scenario 4 (High confounding and high dependent censoring).} 
\begin{itemize}
\item \textbf{Treatment parameters} 
$\gamma_0=1.1,\;\gamma_A=1,\;\gamma_{X1}=-0.5,\;\gamma_{X2}=0.5,\;
\gamma_{t,\mathrm{pre}}=0.4,\;\gamma_{t,\mathrm{post}}=-2.2$
\item \textbf{Survival parameters} 
$\alpha_0=-2.7,\;\alpha_A=-1.0,\;\alpha_{X1}=0.5,\;\alpha_{X2}=-0.8,\;\alpha_{X3}=-0.8,\;
\alpha_{t,\mathrm{pre}}=-0.2,\;\alpha_{t,\mathrm{post}}=0.3$
\item\textbf{Censoring parameters} 
$\beta=0.5,\;\beta_0=-0.3,\;\beta_A=0.3,\;\beta_{X2}=0.5,\;\beta_{X3}=0.5,\;
\beta_{t,\mathrm{pre}}=-0.5,\;\beta_{t,\mathrm{post}}=0.5$
\end{itemize}

\begin{table}[H]
\caption{Performance of estimators (RMST difference in months) with 50 Monte Carlo replicate in scenario 4 with baseline confounders.}
\centering
\label{tab:rmse_baseline}
\begin{tabular}[t]{llrrr}
\toprule
estimator & estimate (sd estimate) & bias & MSE & RMSE\\
\midrule
Naive - Filtered & 99.4 (2.9) & 86.2 & 7441.0 & 86.3\\
G-formula - Filtered & 60.2 (2.5) & 47.0 & 2216.1 & 47.1\\
KM - Cloned & 34.2 (2.4) & 21.1 & 449.6 & 21.2\\
IPCW (A: Logit) - Cloned & 27.6 (2.6) & 14.4 & 214.8 & 14.7\\
G-formula (Cox) - Cloned & 13.0 (1.4) & \textcolor{green_cbf1}{-0.2} & \textcolor{green_cbf1}{2.0} & \textcolor{green_cbf1}{1.4}\\
\addlinespace
G-formula (Weibull) - Cloned & 8.7 (0.8) & -4.4 & 20.3 & 4.5\\
IPCW (A: Logit, N: PW-exp) - Cloned & 12.4 (7.2) & -0.8 & 51.9 & 7.2\\
IPCW (A+N: Logit) - Cloned & 20.5 (3.2) & 7.3 & 63.6 & 8.0\\
IPCW (A+N: PW-exp) - Cloned & 30.0 (2.4) & 16.8 & 287.3 & 17.0\\
\bottomrule
\end{tabular}
\end{table}

\section{Simulation details of Section~\ref{sec-simulation_tv}}\label{sec-coef_tv}


\begin{table}[H]
\centering
\caption{General DGP parameters for the time-dependent confounding setting.}
\label{tab:dgp-timedep-s1}
\renewcommand{\arraystretch}{1.2}
\setlength{\tabcolsep}{6pt}
\small
\begin{tabular}{p{2.8cm} p{11.8cm}}
\toprule
\textbf{Block} & \textbf{Specification / parameters} \\
\midrule
\textbf{General DGP} &
$
X_1 \sim \mathcal N(0.5,1),\;
X_2 \sim \mathcal N(1.0,1),\;
X_{3,0} \sim \mathcal N(1.0,1),\;
X_{4,0} \sim \mathcal N(-1.0,1)
$
\\[4pt]
&
for $k \in (0,1,\dots,K-2)$,
\\[-2pt]
&
$
A_0=1,\quad
\mathbb P(A_{k+1}=1 \mid A_k, X_1, X_2, X_{3,k+1})
=
\operatorname{logit}^{-1}\!\big(
\gamma_{0,\mathrm{long}}+\gamma_{X1}X_1+\gamma_{X2}X_2+\gamma_{X3}(0.1\,k)X_{3,k+1}
+\gamma_{A\mathrm{prev}}A_k+\gamma_t(k)
\big),
$
\\[-2pt]
&
$
\gamma_t(k)=
\begin{cases}
\gamma_{t,\mathrm{pre}}, & k \le k_{\mathrm{switch}},\\
\gamma_{t,\mathrm{post}}, & k > k_{\mathrm{switch}}
\end{cases}
$
\\[6pt]
&
$
X_{3,k+1}
=
\beta_{\mathrm{X3,intercept}}
+\beta_{\mathrm{X3,prev}}X_{3,k}
+\beta_{\mathrm{X3,Aprev}}A_k
+\varepsilon_{k+1},
\qquad
\varepsilon_{k+1}\sim\mathcal N(0,\sigma_{\mathrm{X3}}^2)
$
\\[6pt]
&
$
X_{4,k+1}
=
\epsilon_{\mathrm{X4,intercept}}
+\epsilon_{\mathrm{X4,prev}}X_{4,k}
+\eta_{k+1},
\qquad
\eta_{k+1}\sim\mathcal N(0,\sigma_{\mathrm{X4}}^2)
$
\\[6pt]
&
$
\lambda_T(k \mid A_k,X_1,X_2,X_{3,k},X_{4,k})
=
\exp\!\big(
\alpha_0+\alpha_A A_k+\alpha_{X1}X_1+\alpha_{X2}X_2+\alpha_{X3}X_{3,k}+\alpha_{X4}X_{4,k}+\alpha_t(k)
\big),
$
\\[-2pt]
&
$
\alpha_t(k)=
\begin{cases}
\alpha_{t,\mathrm{pre}}, & k \le 2,\\
\alpha_{t,\mathrm{post}}, & k > 2
\end{cases}
$
\\[6pt]
&
$
\lambda_C(k \mid A_k,X_1,X_2,X_{4,k})
=
\beta\exp\!\big(
\beta_0+\beta_A A_k+\beta_{X1}X_1+\beta_{X2}X_2+\beta_{X4}X_{4,k}+\beta_t(k)
\big),
$
\\[-2pt]
&
$
\beta_t(k)=
\begin{cases}
\beta_{t,\mathrm{pre}}, & k \le 3,\\
\beta_{t,\mathrm{post}}, & k > 3
\end{cases}
$
\\
\midrule

\end{tabular}
\end{table}
\normalsize

\paragraph{Scenario 1 (Low confounding and low dependent censoring).} 
\begin{itemize}
    \item \textbf{Treatment parameters} 
$
\gamma_{0,\mathrm{long}}=-0.10,\;\gamma_{X1}=0.2,\;\gamma_{X2}=-0.2,\;\gamma_{X3}=-0.2,\;
\gamma_{A\mathrm{prev}}=0.65,\;
\gamma_{t,\mathrm{pre}}=0.35,\;\gamma_{t,\mathrm{post}}=-0.40
$
    \item \textbf{Survival parameters} 
$
\alpha_0=-2,\;\alpha_A=-0.75,\;\alpha_{X1}=-0.2,\;\alpha_{X2}=-0.08,\;\alpha_{X3}=-0.08,\;\alpha_{X4}=0.2,\;
\alpha_{t,\mathrm{pre}}=-0.12,\;\alpha_{t,\mathrm{post}}=0.18
$
    \item \textbf{Censoring parameters} 
$
\beta=0.10,\;\beta_0=-1.5,\;\beta_A=0,\;\beta_{X1}=0.2,\;\beta_{X2}=-0.03,\;\beta_{X4}=-0.5,\;
\beta_{t,\mathrm{pre}}=-0.35,\;\beta_{t,\mathrm{post}}=-0.05
$
    \item \textbf{Time-varying covariate parameters} 
$
\beta_{\mathrm{X3,intercept}}=0.05,\;\beta_{\mathrm{X3,prev}}=0.15,\;\beta_{\mathrm{X3,Aprev}}=0.08,\;\sigma_{\mathrm{X3}}=0.22
$
\\

$
\epsilon_{\mathrm{X4,intercept}}=0.05,\;\epsilon_{\mathrm{X4,prev}}=0.15,\;\sigma_{\mathrm{X4}}=0.22
$
\end{itemize}

\paragraph{Scenario 2 (High confounding and low dependent censoring).} 
\begin{itemize}
    \item \textbf{Treatment parameters} 
$
\gamma_{0,\mathrm{long}}=-0.10,\;\gamma_{X1}=0.05,\;\gamma_{X2}=-0.05,\;\gamma_{X3}=-1,\;
\gamma_{A\mathrm{prev}}=0.62,\;
\gamma_{t,\mathrm{pre}}=0.24,\;\gamma_{t,\mathrm{post}}=-0.50
$

    \item \textbf{Survival parameters} 
$
\alpha_0=-3.2,\;\alpha_A=-1.35,\;\alpha_{X1}=-0.30,\;\alpha_{X2}=0.30,\;\alpha_{X3}=1,\;\alpha_{X4}=0.70,\;
\alpha_{t,\mathrm{pre}}=-0.08,\;\alpha_{t,\mathrm{post}}=0.35
$
    \item \textbf{Censoring parameters} 
$
\beta=0.08,\;\beta_0=-1.0,\;\beta_A=0,\;\beta_{X1}=0.02,\;\beta_{X2}=-0.02,\;\beta_{X4}=0.08,\;
\beta_{t,\mathrm{pre}}=-0.45,\;\beta_{t,\mathrm{post}}=-0.05
$
    \item \textbf{Time-varying covariate parameters} 
$
\beta_{\mathrm{X3,intercept}}=0,\;\beta_{\mathrm{X3,prev}}=0.78,\;\beta_{\mathrm{X3,Aprev}}=0.50,\;\sigma_{\mathrm{X3}}=0.18
$
\\

$
\epsilon_{\mathrm{X4,intercept}}=0,\;\epsilon_{\mathrm{X4,prev}}=0.50,\;\sigma_{\mathrm{X4}}=0.22
$
\end{itemize}

\paragraph{Scenario 3 (Low confounding and high dependent censoring).} 
\begin{itemize}
    \item \textbf{Treatment parameters} 
$
\gamma_{0,\mathrm{long}}=-0.55,\;\gamma_{X1}=0.08,\;\gamma_{X2}=-0.08,\;\gamma_{X3}=0.20,\;
\gamma_{A\mathrm{prev}}=1.00,\;
\gamma_{t,\mathrm{pre}}=1.00,\;\gamma_{t,\mathrm{post}}=-0.1
$

    \item \textbf{Survival parameters} 
$
\alpha_0=-3.20,\;\alpha_A=-1.00,\;\alpha_{X1}=-0.15,\;\alpha_{X2}=0.15,\;\alpha_{X3}=0.08,\;\alpha_{X4}=1.10,\;
\alpha_{t,\mathrm{pre}}=-0.15,\;\alpha_{t,\mathrm{post}}=0.20
$
    \item \textbf{Censoring parameters} 
$
\beta=0.26,\;\beta_0=-1.05,\;\beta_A=0.20,\;\beta_{X1}=0.20,\;\beta_{X2}=-0.20,\;\beta_{X4}=1.45,\;
\beta_{t,\mathrm{pre}}=0.45,\;\beta_{t,\mathrm{post}}=0.10
$
    \item \textbf{Time-varying covariate parameters} 
$
\beta_{\mathrm{X3,intercept}}=0.05,\;\beta_{\mathrm{X3,prev}}=0.70,\;\beta_{\mathrm{X3,Aprev}}=0.10,\;\sigma_{\mathrm{X3}}=0.30
$
\\

$
\epsilon_{\mathrm{X4,intercept}}=0,\;\epsilon_{\mathrm{X4,prev}}=0.92,\;\sigma_{\mathrm{X4}}=0.25
$
\end{itemize}

\paragraph{Scenario 4 (High confounding and high dependent censoring).} 
\begin{itemize}
    \item \textbf{Treatment parameters} 
$
\gamma_{0,\mathrm{long}}=-1.15,\;\gamma_{X1}=0.30,\;\gamma_{X2}=-0.30,\;\gamma_{X3}=2.5,\;
\gamma_{A\mathrm{prev}}=1.2,\;
\gamma_{t,\mathrm{pre}}=1.40,\;\gamma_{t,\mathrm{post}}=-3.00
$

    \item \textbf{Survival parameters} 
$
\alpha_0=-4,\;\alpha_A=-1.10,\;\alpha_{X1}=-0.35,\;\alpha_{X2}=0.35,\;\alpha_{X3}=1.10,\;\alpha_{X4}=0.85,\;
\alpha_{t,\mathrm{pre}}=-0.10,\;\alpha_{t,\mathrm{post}}=0.35
$
    \item \textbf{Censoring parameters} 
$
\beta=0.28,\;\beta_0=-0.85,\;\beta_A=0.75,\;\beta_{X1}=0.70,\;\beta_{X2}=-0.90,\;\beta_{X4}=1,\;
\beta_{t,\mathrm{pre}}=0.45,\;\beta_{t,\mathrm{post}}=0.05
$
    \item \textbf{Time-varying covariate parameters} 
$
\beta_{\mathrm{X3,intercept}}=0.10,\;\beta_{\mathrm{X3,prev}}=0.85,\;\beta_{\mathrm{X3,Aprev}}=0.85,\;\sigma_{\mathrm{X3}}=0.25
$
\\

$
\epsilon_{\mathrm{X4,intercept}}=0,\;\epsilon_{\mathrm{X4,prev}}=0.85,\;\sigma_{\mathrm{X4}}=0.35
$
\end{itemize}

\begin{table}[H]
\caption{Performance of estimators (RMST difference in months) with 50 Monte Carlo replicate in scenario 4 with time-varying confounders.}
\centering
\label{tab:rmse_tv}
\begin{tabular}[t]{llrrr}
\toprule
estimator & estimate (sd estimate) & bias & MSE & RMSE\\
\midrule
Naive - Filtered & 53.0 (3.8) & 28.1 & 801.7 & 28.3\\
G-formula - Filtered & 27.7 (1.4) & 2.7 & 9.4 & 3.1\\
KM - Cloned & 12.2 (3.1) & -12.7 & 171.6 & 13.1\\
IPCW (A: Logit) - Cloned & 10.1 (2.9) & -14.8 & 228.0 & 15.1\\
G-formula - Cloned & 74.8 (4.6) & 49.9 & 2507.6 & 50.1\\
\addlinespace
IPCW (A: Logit, N: PW-exp) - Cloned & 29.5 (13.1) & 4.6 & 188.1 & 13.7\\
IPCW static only (A: Logit, N: PW-exp) - Cloned & 40.8 (31.6) & 15.8 & 1225.4 & 35.0\\
IPCW (A+N: Logit) - Cloned & 29.8 (10.2) & 4.9 & 125.8 & 11.2\\
IPCW (A+N: PW-exp) - Cloned & 34.4 (5.8) & 9.4 & 121.4 & 11.0\\
\bottomrule
\end{tabular}
\end{table}

\end{appendix}

\end{document}